\documentclass{article}
\usepackage{amsmath}
\usepackage{amssymb}

\usepackage{graphicx}
\usepackage{float}
\restylefloat*{figure}

\usepackage{subfig} 
\usepackage{epsfig}
\usepackage{caption}
\usepackage{subcaption}
\usepackage{etoolbox}
\usepackage{amsthm}
\usepackage{ytableau}
\usepackage{tikz}
\usepackage{makecell}
\usepackage[utf8]{inputenc}
\usepackage{multirow}
\usepackage{xfrac}
\usepackage{upgreek}
\usepackage{dsfont} 
\usepackage{empheq}

\usepackage[colorlinks,linkcolor=blue,citecolor=blue]{hyperref}
\usepackage{aliascnt}
\newtheorem{theorem}{Theorem}[section]
\newaliascnt{lemma}{theorem}
\newtheorem{lemma}[lemma]{Lemma}
\aliascntresetthe{lemma}

\newaliascnt{corollary}{theorem}
\newtheorem{corollary}[corollary]{Corollary}
\aliascntresetthe{corollary}

\theoremstyle{definition}

\newaliascnt{definition}{theorem}
\newtheorem{definition}[definition]{Definition}
\aliascntresetthe{definition}

\newaliascnt{example}{theorem}
\newtheorem{example}[example]{Example}
\aliascntresetthe{example}

\newaliascnt{remark}{theorem}
\newtheorem{remark}[remark]{Remark}
\aliascntresetthe{remark}

\newaliascnt{assumption}{theorem}

\aliascntresetthe{assumption}

\newaliascnt{proposition}{theorem}
\newtheorem{proposition}[proposition]{Proposition}
\aliascntresetthe{proposition}

\usepackage[a4paper,margin=2.5cm,footskip=0.25in]{geometry}

\DeclareMathOperator{\sgn}{sgn}
\DeclareMathOperator{\tr}{tr}

\counterwithin*{equation}{section}
\counterwithin*{equation}{section}

\newcommand{\beq}{\begin{equation}}  
\newcommand{\eeq}{\end{equation}}  
\newcommand\dd{\mathrm{d}} 
\newcommand\bm{\mathbf{m}}
\newcommand\bv{\mathbf{v}}
\newcommand\bV{\mathbf{V}}
\newcommand\bF{\mathbf{F}}
\newcommand\bG{\mathbf{G}}
\newcommand\bL{\mathbf{L}}
\newcommand\bM{\mathbf{M}}
\newcommand\Cl{\mathcal{C}\ell}

\newcommand{\F}{{\mathbb F}}

\newcommand{\C}{{\mathbb C}}
\newcommand{\R}{{\mathbb R}}

\newcommand{\rd}{\mathrm{d}}

\newcommand{\ri}{\mathrm{i}}

\newcommand\la{{\lambda}}

\newcommand\al{{\alpha}}

\newcommand\ze{{\zeta}}

\newcommand\Jac{\mathrm{Jac}\,}

\newcommand{\si}{\sigma}

 \newcommand{\avg}[1]{\bigl\langle #1 \bigr\rangle}
\newcommand{\abs}[1]{\left\lvert #1 \right\rvert}
\title{Vector peakon equations and isospectral flows in Clifford algebras}

\author{ A.N.W. Hone\footnote{Corresponding author e-mail: A.N.W.Hone@kent.ac.uk}~\\
School of Engineering, Mathematics \& Physics~\\
    University of Kent, Canterbury CT2 7NF, UK\\ \\
    V.S. Novikov\\ Department of Mathematical Sciences\\ Loughborough University, LE11 3TU, UK \\ \\
    J. Szmigielski\thanks{e-mail: szmigiel@math.usask.ca}~\\ 
    Department of Mathematics \& Statistics~\\
    and Centre for Quantum Topology and Its Applications (quanTA)~\\ University of Saskatchewan, Saskatoon S7N 5E6, SK, CANADA }

\begin{document}

\maketitle

\begin{abstract} 
Starting from a spectral problem posed in a Clifford algebra with $d$ generators and Euclidean signature, we study an integrable, coupled  system of PDEs that can be viewed as a vector perturbation of the Camassa--Holm equation with residual orthogonal symmetry.    
In the two-component case $d=2$, we show that the travelling wave solutions correspond to a Liouville integrable Hamiltonian system with two degrees of freedom, making use of a reciprocal transformation linking the coupled PDEs to a symmetry of the Hirota--Satsuma system. We also present a symmetry  classification of all integrable two-component  perturbations of Camassa--Holm, and find that besides the $d=2$ system analyzed here, the coupled 2CH 
system studied by Olver and Rosenau   (as well as by Chen, Liu and Zhang, and Falqui), and equations related to either of those systems by Miura transformations, we also obtain a new system that (to the best of our knowledge) has not been reported previously.  
For the case of an arbitrary number of components $d$, we additionally investigate the short-pulse (high-frequency) regime, in which the limiting dynamics are governed by a vector-valued Hunter–Saxton-type system. Furthermore, we provide a detailed analysis of the corresponding measure-valued (weak) solutions associated with this system.
\end{abstract} 

\section{Introduction} \label{sec:Intro}
\setcounter{equation}{0}
In this paper we study the system of partial differential equations
\beq\label{vecsys}
\bm_t = 2u_x \, \bm + u\, \bm_x + \bV\, \bm, 
\eeq 
where $\bm=\bm(x,t)$ is a $d$-component vector and $\bV=(v_{\mu\nu})$ is a sparse skew-symmetric $d\times d$ matrix of the form
\beq\label{Vdef}
\bV = \left( \begin{array}{cc}
0 & \bv^T \\
-\bv & \mathbf{0}
\end{array}
\right),
\eeq
with $\bv=(v_{1\nu})_{\nu=2,\ldots,d}$. The scalar field $u$ and the components of $\bv$ are related to the components of $\bm=(m_\mu)_{\mu=1,\ldots,d}$ through the constraints
\beq\label{constraints}
u_{xx}-4\al u = 2m_1, \qquad \mathrm{and} \qquad 
v_{1\mu, x}=-2 m_{\mu}\quad \mathrm{for}\,\, \mu=2,\ldots, d, 
\eeq 
where $\al$ is an arbitrary constant.  Writing \eqref{vecsys} in components yields
\beq\label{eqcpts}
\begin{array}{rcl}
m_{1,t} & = & 2 u_x m_1 +um_{1,x}+ \sum_{\nu=2}^n v_{1\nu}m_\nu,  \\
m_{\mu,t} & = &  2u_x m_\mu +um_{\mu,x} -  v_{1\mu}m_1, \qquad \qquad 
\mu=2,\ldots,d. 
\end{array}
\eeq 

The system \eqref{vecsys} is a $d$-dimensional Euclidean analogue of the two-component system introduced in \cite{BS2022}, which arose from the Euler--Bernoulli beam equation and was subsequently extended to a Clifford algebra framework in \cite{BS2025}. A preliminary analysis of two-component Camassa--Holm dynamics with Minkowski signature can also be found in \cite{BS2025}.  The present paper focuses on the corresponding setting with Euclidean signature.  In addition, we place \eqref{vecsys} within the broader framework of the symmetry approach to integrability: in particular, for $d=2$ we obtain a complete classification of a well-defined class of integrable perturbations of the scalar Camassa--Holm equation within this vector/Clifford setting.

In the case when $\bv=0$, when $m_\mu=0$ for $2\leq \mu\leq d$, the system \eqref{vecsys} reduces to the standard  (scalar) Camassa--Holm equation, which was derived as an asymptotic model of shallow water waves in \cite{CH1993}. As well as being completely integrable, the Camassa--Holm equation was found to combine many remarkable features, including smooth solitons, wave breaking, and an interpretation as an Euler-Poincar\'e flow on a diffeomorphism group, as well as a new class of weak soliton solutions known as peakons. The surprising combination of all these properties has led to many attempts to generalize from the scalar case to multi-component analogues of the Camassa--Holm equation.
A complex NLS-type analogue was presented by Fokas in \cite{fokas}, based on the methodology of 
bi-Hamiltonian structures and recursion operators developed previously with Fuchssteiner, 
while a few different 
integrable two-component examples were 
introduced via a tri-Hamiltonian approach by Olver and Rosenau \cite{or}. One of the latter examples, commonly referred to as 2CH, was considered in more detail by Chen, Liu and Zhang \cite{CLZ, LZ}, and also Falqui \cite{falqui}, who used a bi-Hamiltonian structure of hydrodynamic type as the basis for their analysis, generalizing the known structure of the original scalar CH equation. Another approach based on supersymmetry, due to Popowicz \cite{pop}, produced non-integrable equations with only one Hamiltonian structure.

Since that time, multi-component analogues of Camassa--Holm have been constructed from  multiple different viewpoints, usually by starting from an appropriate family of Lax operators \cite{HI, lifeng,  XQ2016, XZQ2015},  from associated geometric flows \cite{qsy}, or as particular cases of a larger set of equations \cite{bkm}. For a restricted class of two-component systems, there is also a classification of integrable coupled equations of 
Camassa--Holm type \cite{HNW}. However, the literature on such equations is now unfeasibly large, so the small selection of papers cited here is necessarily incomplete. 

Given the plethora of different multi-component equations of Camassa--Holm type, it is worth pointing out some of the distinctive features of the system \eqref{vecsys}.  This system 
equips the scalar Camassa--Holm with internal degrees of freedom, represented by the skew-symmetric tensor ${\bf V}$, similar to an orbital momentum or spin.
The global symmetry of the system is $O(d-1)$, under which the tensor ${\bf V}$ transforms into another tensor of the same type. 
The Clifford algebra structure creates complex interactions between internal degrees of freedom. In the simplest case of peakon solutions, this Clifford structure influences the long-time behaviour, leading to coordinated energy exchange even when the peakons are spatially separated.

To keep the presentation self-contained, we first review (and then extend) the setup of \cite{BS2025}.  As part of this review we explain in what sense the vector PDE \eqref{vecsys} (equivalently, the component system \eqref{eqcpts}) is integrable, by connecting it to an isospectral deformation of a linear problem formulated in the Clifford algebra $\mathcal{C}\ell(W)$, where $W$ is a $d$-dimensional vector space equipped with a nondegenerate symmetric bilinear form.  We show that for $d=2$ the system \eqref{vecsys} admits an infinite hierarchy of commuting symmetries.  For arbitrary $d$, we subsequently specialize to $\al=0$, corresponding to the short-pulse (Hunter--Saxton) limit \cite{HS, HZ}, and construct weak solutions of \eqref{vecsys} arising from measure-valued elements of $\mathcal{C}\ell(W)$.

For comparison, the beam system considered in \cite{BS2022} is
\begin{align} \label{bs} 
m_t&=(um)_x+u_xm-vm, \\
n_t&=(un)_x+u_xn+vn, \label{bsan}
\end{align}
with constraints
\begin{equation} \label{eq:bs-constraints} 
u_{xxx}-4 u_x=(m+n)_x, \quad v_x=(n-m), 
\end{equation} 
corresponding to setting $\al=1$ in 
\eqref{constraints}. 
The underlying principle from which equations \eqref{bs}-\eqref{eq:bs-constraints} were derived was that of isospectral deformations of the Euler-Bernoulli beam problem.  
The same system of equations also appears in the list of three interesting integrable equations 
connected to a $4\times 4$ matrix Lax pair in the work of Geng and Wu \cite{GW}.

Setting $m=m_1+m_2$ and $n=m_1-m_2$, the system consisting of \eqref{bs} and \eqref{bsan} becomes
\begin{align*}
m_{1,t}&=(um_1)_x+u_xm_1-vm_2, \\
m_{2,t}&=(um_2)_x+u_xm_2-vm_1, 
\end{align*} 
together with
$$ 
u_{xxx}-4u_x=2m_{1,x}, \qquad v_x=2m_2.
$$ 
This system is associated with a Clifford algebra of Minkowski signature $(+,-)$.  By contrast, the Euclidean signature $(+,+)$ leads (with a minor modification in the constraint, involving the parameter $\alpha$ used throughout this work) to
\begin{align}
\label{2cptsys}
&\left\{ \begin{array}{rl}
m_{1,t}&=(um_1)_x+u_xm_1+vm_2, \\
m_{2,t}&=(um_2)_x+u_xm_2-vm_1, 
\end{array}\right. 
\end{align}
with constraints
\begin{equation}
u_{xxx}-4\alpha u_x=2m_{1,x}, \qquad v_x=-2m_2.
\end{equation}
In the second part of the paper, we concentrate on the Hunter--Saxton limit $\alpha=0$.  Moreover, we assume that the measures $m_1,m_2$ are compactly supported and work with the integrated form of the first constraint, namely $u_{xx}=2m_1$.
For a fixed value of $\alpha\neq 0$, the Hunter--Saxton analogue of the system corresponds to taking 
the short pulse, high frequency limit of (\ref{2cptsys}). Upon rescaling the dependent and independent variables according to 
 $$
 m_1\to \epsilon^{-2} m_1, \quad 
 m_2\to \epsilon^{-2} m_2,  \quad u\to u, \quad v \to \epsilon^{-1} v, \quad 
 x\to \epsilon x, \quad t \to \epsilon t, 
 $$
 and taking the limit $\epsilon\to 0$,
 we find that the form of the system (\ref{2cptsys}) is preserved, and only the relation between $m_1$ and $u$ is altered, so that after integrating the first relation we have: 
 \beq \label{scaled}
 2m_1=u_{xx}, \qquad 2m_2=-v_x.
 \eeq 
 In section 5 below, we will be concerned with constructing weak solutions of the $d$-component version of the short pulse limit. 

The overall structure of the paper is as follows. In the next section, we briefly review the Clifford algebra structure which provides the setting for the spectral problem and Lax pair of the $d$-component system (\ref{vecsys}), before explaining how symmetries and Hamiltonian operators appear in this context. In section 3, the main focus is on the case of two components ($d=2$). We present a reciprocal transformation, mapping the conservation laws of the two-component PDE system to those of another system with a new set of independent variables. Under this transformation, we demonstrate that the problem of finding travelling waves for this system can be reduced to that of solving a finite-dimensional Hamiltonian system with two degrees of freedom that is Liouville integrable, which is then used to produce  plots of travelling wave profiles, obtained via numerical integration.  The conclusion of the third section is devoted to explaining how, when $d=2$, the first higher symmetry of the PDE system is connected by the reciprocal transformation to another coupled pair of PDEs that admits a Miura map to the Hirota-Satsuma system. Hence, under the latter transformation, the symmetries of the original two-component system are related to the Hirota-Satsuma hierarchy. This connection prompted the subsequent investigations whose results are presented in section 4, where we present a complete classification of two-component deformations of the Camassa-Holm equation via the symmetry approach. The result of the analysis is a list 
of six different systems, all but one of which turn out to be related to known integrable hierarchies by a combination of reciprocal and Miura transformations, while the one remaining system appears to be new. In section 5, we consider the Hunter-Saxton limit $\al=0$ of the full $d$-component system, and use Clifford-valued continued fractions to solve the Neumann problem for the string equation with a discrete measure $M$. This leads to weak soliton solutions of the vector which are piecewise linear, with moving peaks in the scalar field $u$ at the support on the measure, hence the name peakons. For the two-component case $d=2$, we analyze the dynamics of two peakons in more detail, and plot an explicit solution of the equations of motion of the peaks, including additional internal degrees of freedom. The paper ends with some conclusions and open problems.  

Before proceeding with the rest of the paper, there are a couple of things that are worth pointing out. Firstly, unlike the original Camassa-Holm peakons, which have the shape of the Green's function  $e^{-|x|}$ for the 1D Helmholtz operator, here we are considering the  Hunter-Saxton (or short pulse)  limit, in which case the relevant function is $|x|$, and the peakons are piecewise linear. Secondly, 
the main purpose of the paper is to introduce and analyze a vector generalization of the Camassa–Holm equation, which is a well-studied nonlinear shallow-water equation with a rich mathematical and physical history; but as it stands, we do not have a specific physical interpretation of the vector system.  
Thus, we would like to  emphasize that the work should be read as a mathematical study of an integrable vector perturbation of the Camassa–Holm equation, rather than as a direct derivation of a new physical model.

\section{Lax pair and higher symmetries}\label{laxsym}
 \setcounter{equation}{0}
 We will now briefly review the setup of \cite{BS2025}.  
 The system (\ref{vecsys}) arises from a Lax pair defined in a Clifford algebra $\Cl(W)$, where $W$ is a $d$-dimensional vector space over a field $\F$ 
 with a symmetric bilinear form $(\, , \, )$. 
 We usually assume $\F=\C$, but sometimes 
 we work with $\F=\R$, and consider different real forms of 
 the underlying complex Clifford algebra.  
 Upon choosing an orthonormal basis of generators $e_\mu$, $\mu=1,\ldots, d$ 
 such that 
 $(e_\mu,e_\nu)=\varepsilon_\mu\delta_{\mu \nu}$, 
 for signs $\varepsilon_\mu=\pm 1$,  the multiplication in 
 $\Cl(W)$ is then given by 
 $$ 
 e_\mu e_\nu + e_\nu e_\mu = 2\varepsilon_\mu\delta_{\mu\nu} \, 1 , 
 $$
 where $1$ denotes the unit $1\in\F\subset\Cl (W)$, 
 which will usually be omitted from formulae. 
 The signs $\varepsilon_\mu$ determine the signature of the 
 bilinear form $(\, ,\,)$: over $\C$ one can always fix 
 $\varepsilon_\mu=1$ for all $\mu$, which is 
 our canonical choice for most of the paper, but on occasions when 
 real solutions are being considered we 
 can 
 allow  
 different signatures over $\R$. 
  The whole algebra is a direct sum 
 of graded components: 
 $$ 
 \Cl (W) = \Cl_0\oplus \Cl_1\oplus\cdots \oplus \Cl_n,  
 $$
 where $\Cl_0=<1>\cong \F$, $\Cl_1 =<e_\mu>_{\mu=1,\ldots, n}\cong W$, 
 $\Cl_2 =<e_\mu e_\nu>_{\mu<\nu}$, etc. We can also 
 equip $\Cl(W)$ with the trace $\mathrm{tr}$, a linear functional 
 on the Clifford algebra defined by $\mathrm{tr} (1)=1$ 
 and $\mathrm{tr} (c)=0$ for all 
 $c\in\Cl_j$ with $j\geq 1$, which then has the property 
 that 
 $$ 
 \mathrm{tr}(ab)=(a,b) \qquad \forall\,  a,b\in W\cong \Cl_1. 
 $$
 
 To construct a vector perturbation of the Camassa--Holm system, 
 we start from a linear 
 problem based on 
 a string equation with a potential 
 $$M=\sum_{\mu=1}^nm_\mu\, e_\mu\in\Cl_1,$$ that is 
 \beq\label{Phixx}
D_x^2 \Phi = (\al   
+ \la M)\Phi, 
\eeq 
where the wave function $\Phi\in \Cl(V)$, with an arbitrary scalar constant $\al\in\F$ and spectral parameter $\la$, 
and couple this to a time evolution written as 
\beq\label{Phit}
D_t \Phi = \left(\tfrac{1}{2}(a-b_x) + bD_x\right) \Phi.
\eeq
Following \cite{BS2022, BS2025}, we refer to 
the linear equation \eqref{Phixx} as a (generalized) beam problem. 
\begin{lemma}[\cite{BS2025}]
\label{comlem}
The compatibility conditions for the Lax pair given by 
(\ref{Phixx}) and (\ref{Phit}) consist of the following two relations: 
\beq\label{compat}
\begin{array}{rcl} 
a_x + \la [b,M]& = & 0, \\ 
\la M_t  & = & 2\al b_x -\tfrac{1}{2}b_{xxx} + \tfrac{\la}{2} 
\left([a,M] + b\,  \overset{\leftarrow}{\cal L}_M + 
\overset{\rightarrow}{\cal L}_M \, b\right) ; 
\end{array} 
\eeq 
in the above, ${\cal L}_M$
is the operator $MD_x +D_x M$, and the arrow on top denotes whether it acts to the right or the left.
\end{lemma}
\begin{remark} To clarify the notation as used in 
\cite{BS2025}, note that we have 
$b\,  \overset{\leftarrow}{\cal L}_M = 2b_x M + bM_x$ and $\overset{\rightarrow}{\cal L}_M \, b=2Mb_x+M_xb$. Thus the sum of these two terms reveals the presence of 
anticommutators in the compatibility conditions, as well as the commutators that are apparent in \eqref{compat}. 
\end{remark}
The equation (\ref{vecsys}) is an isospectral evolution, 
corresponding to a particular solution of the compatibility 
conditions (\ref{compat}), obtained by choosing 
$$
a=a_0\in\Cl_2, \qquad b = b_0 +\frac{b_{-1}}{\la}\quad\mathrm{with}
\quad b_0 \in \Cl_0,\,\, b_{-1}\in \Cl_1, 
$$
and this gives a consistent result provided that $b_{-1}$ is constant. Without loss of generality, we can always perform a  gauge transformation 
by multiplying $\Phi$ in the linear system by a constant element (corresponding to a rotation in $SO(d)$), so that up to scale 
we can fix $b_{-1}=e_1$. Then we find 
\beq\label{coeffs}
b_0 = u \in\F 
, \qquad b_{-1}=e_1\in W, \qquad a_0 
=\sum_{\nu=2}^d v_{1\nu}\, e_1 e_\nu \in \Cl_2, 
\eeq 
where $u$ and the coefficients appearing in 
$a_0$ 
are subject to the constraints \eqref{constraints}.

We conclude this section with several remarks regarding the place of \eqref{vecsys} in a larger context of integrable systems.  
We note that we can rewrite \eqref{vecsys} as an evolution equation in the Clifford 
algebra, namely 
\beq\label{mtcliff}
M_t =\tfrac{1}{2}(A_M D_x + D_x A_M) \, 
u 
-\tfrac{1}{2}C_M D_x^{-1} C_M \, e_1, 
\eeq 
with $u=b_0,e_1=b_{-1}$ as in \eqref{coeffs}, 
where $A_M,C_M$ respectively denote the anticommutator/commutator with respect to $M$. 

A pair of elementary symmetries is produced by the choices 
$$ 
a=a_0 \in\Cl_2, \qquad b=b_0 \in \Cl_0, \qquad a_0,b_0\quad\mathrm{const} , 
$$
for where the constants  $a_0,b_0$ are arbitrary elements 
of degree 2 and 0, respectively, and the corresponding symmetries 
can be written in combination as the $s$-flow
$$ 
M_s =\tfrac{1}{2}[a_0,M]+ b_0 M_x, 
$$
while the first non-trivial symmetry is 3rd order, being given by the flow  
\beq\label{3rdorderflow} 
M_\tau =\tfrac{1}{2} (4\al D_x - D_x^3) \, \left(|\bm|^{-\frac{3}{2}} \, M \right).  
\eeq 
The latter is Hamiltonian, as it can be expressed as 
$$ 
M_\tau = {\cal B}_1 \frac{\delta H}{\delta M}, $$
with the Hamiltonian operator 
\begin{equation} \label{eq:B1}
{\cal B}_1 = (4\al D_x - D_x^3) 
\end{equation}
and conserved functional 
\beq\label{fnal}
H =\int (\mathrm{tr}\, M^2)^{\frac{1}{4}}\, \dd x. 
\eeq 
The conserved density associated with $H$ is 
$$ 
\rho = (\mathrm{tr}\, M^2)^{\frac{1}{4}} = 
\Big(\sum_{\mu=1}^d m_\mu^2 \Big)^{\frac{1}{4}}= |\bm|^{\frac{1}{2}}, 
$$
which for $d=1$ recovers 
the well-known density $\sqrt{m}$ for the Camassa--Holm equation. 
The gradient of the Hamiltonian $H$ is in the kernel of the 
second Hamiltonian structure ${\cal B}_2$, which is compatible with 
${\cal B}_1$, being given by the operator 
\begin{equation} \label{eq:B2}
{\cal B}_2 = 
(A_M D_x + D_x A_M) +C_M D_x^{-1} C_M 
\end{equation}
(cf.\ \eqref{mtcliff} above). 
These operators were found in the work of Olver and Sokolov, who 
considered integrable Hamiltonian systems defined in associative algebras 
\cite{os}.

 \section{Two-component case: reciprocal transformation, 
 travelling waves and Hirota--Satsuma} 
\setcounter{equation}{0}

In this section we focus on 
the case $d=2$ of  \eqref{eqcpts}, which is the two-component system
\beq\label{2cpt}
\begin{array}{rcl}
    m_{1,t} & = & 2 u_x m_1 +um_{1,x}+ vm_2,  \\
      m_{2,t} & = &  2u_x m_2 +um_{2,x} -  v m_1, \qquad 
\end{array}
\eeq 
with auxiliary scalar fields $u,v$ being related to $m_1,m_2$ by 
the constraints 
\beq\label{uveqs}
u_{xx}-u =2m_1, \qquad v_x = -2m_2. 
\eeq
(Note that we fixed $\al=1/4$ here.) We have written \eqref{2cpt} in the case of Euclidean signature. As previously mentioned in the introduction, up to a linear transformation and rescaling, this two-component system is equivalent to the beam system introduced by Beals and Szmigielski in \cite{BS2022}. 

To begin with, we present a reciprocal transformation, which changes the conservation laws in the independent variables $x,t$ into conservation laws involving new variables $X,T$.   
We discuss how this transforms the linear system (\ref{Phixx}) and (\ref{Phit}) for $d=2$, thereby producing a new Lax pair in terms of  the  new variables $X,T$. 

Next, we proceed to consider the most elementary solutions of \eqref{2cpt}, namely travelling waves: after proving an elementary result on the non-existence of such waves with vanishing boundary conditions, we present a general lemma on the correspondence between travelling waves of different PDE systems related to one another by a reciprocal transformation. 
We then use this result to produce a hodograph link between the travelling waves of \eqref{2cpt} with a system of Hamiltonian ODEs with two degrees of freedom, and prove that this is integrable in the Liouville sense. 

In a subsequent subsection, we describe the connection between 
\eqref{2cpt} and the hierarchy of 
the Hirota--Satsuma equation, 
a well-known integrable PDE. 
This leads onto the following section, where we proceed to present a classification of all two-component deformations of the Camassa--Holm equation, including \eqref{2cpt}, which possess a higher local symmetry whose undeformed version has order 3. 

\subsection{Reciprocal transformation and Lax pair}\label{sec:RT}

The linear system consisting of \eqref{Phixx} and \eqref{Phit} can be written in the first order form 
\beq\label{FGlin}
\Psi_x = \bF\Psi, \qquad \Psi_t=\bG\Psi, 
\eeq 
where 
$$ 
\Psi = \left(\begin{array}{c} \Phi\,\, \\ 
\,\,\,\Phi_x\,\,\end{array}
\right) 
$$
and $F,G$ are $2\times 2$ matrices with entries in 
$\Cl (W)$, namely 
\beq\label{FGeqs}
\bF = \left( \begin{array}{cc}  0 & 1 \\ 
\al 1+\la M & 0 
\end{array}
\right) , \qquad 
\bG = \left( \begin{array}{cc}  \frac{1}{2}(a-b_x) & b \\ 
\frac{1}{2}(a_x-b_{xx})+\al b +\la bM & \frac{1}{2}(a+b_x)
\end{array}
\right) . 
\eeq 
The compatibility condition for 
\eqref{FGlin} is the zero curvature equation 
$$
\bF_t-\bG_x+[\bF,\bG]=0, 
$$
which produces the same set of 
conditions \eqref{compat} as in Lemma \ref{comlem}.

For what follows, in the case $d=2$ it is convenient to pick a specific $2\times 2$ matrix representation for the Clifford algebra $\Cl (W)$ with generators $e_1.e_2$. This is a 4-dimensional algebra spanned by 
$1,e_1,e_2,e_1e_2$, and by fixing particular representatives for the generators we may take  
$$ 
1\rightarrow \left(\begin{array}{cc} 1 & 0 \\ 
0 & 1
\end{array}\right), \quad 
e_1 \rightarrow \left(\begin{array}{cc} 0 & 1 \\ 
1 & 0
\end{array}\right), \quad e_2 \rightarrow \left(\begin{array}{cc} 0 & -\ri \\ 
\ri & 0
\end{array}\right) ,  
\quad 
e_1e_2 \rightarrow \left(\begin{array}{cc} \ri & 0 \\ 
0 & -\ri
\end{array}\right), 
$$
resulting in a $4\times 4$ representation of the 
linear system \eqref{FGlin}. Upon taking 
$$M = m_1e_1+m_2 e_2, \qquad a=v e_1e_2, \qquad 
b=u + \la^{-1}e_1, 
$$ 
in accordance with \eqref{coeffs} when $d=2$, and fixing $\al=\frac{1}{4}$, we obtain the 
linear system for \eqref{2cpt}  rewritten in the 1st order form \eqref{FGlin},
with the $4\times 4$ matrices $\bF,\bG$ 
built from the appropriate $2\times 2$ blocks, 
given by 
\beq \label{FGforms} \begin{array}{rcl}
\bF& = &\left( \begin{array}{cccc}  0 & 0 & 1 & 0 \\
0 & 0 & 0 & 1  \\
\frac{1}{4} & \la (m_1-\ri m_2) & 0 & 0 \\
\la (m_1+\ri m_2) & \frac{1}{4} & 0 & 0 
\end{array}
\right) , \\ \\
\bG & = & \left( \begin{array}{cccc}  -\frac{1}{2}(u_x-\ri v) & 0 & u & \la^{-1}  \\ 
0 &  -\frac{1}{2}(u_x+\ri v)  & \la^{-1} & u  \\
 -\frac{1}{4}u & \la u(m_1-\ri m_2)+\frac{1}{4}\la^{-1} & \frac{1}{2}(u_x+\ri v) & 0  \\
\la u(m_1+\ri m_2)+\frac{1}{4}\la^{-1} &  -\frac{1}{4}u  & 0 &  \frac{1}{2}(u_x-\ri v)   
\end{array}
\right) .
\end{array}
\eeq 

Our main goal in this section is to transform the partial differential equations (\ref{2cpt}) and the associated linear system \eqref{FGlin} by means of the reciprocal transformation 
\beq\label{RT}
\rd X =\rho\, \rd x + \rho u\, \rd t, \qquad \rd T=\rd t, 
\eeq 
where $\rho$ is the density associated with the functional (\ref{fnal}), that is 
\beq\label{rhodefn2}
\rho = |\bm|^{\frac{1}{2}}=\big( m_1^2 +m_2^2 \big)^\frac{1}{4}
\eeq 
when $d=2$. This means that  derivatives with respect to the independent variables transform according to  
$$ 
D_x = \rho\, D_X, \qquad D_t =D_T +\rho u\, D_X.
$$
The consistency of the reciprocal transformation is guaranteed by the conservation law 
$$ 
\rho_t = (\rho u)_x, 
$$ which transforms under \eqref{RT} to another conservation law in the new independent variables $X,T$, namely 
\beq\label{rhoinv} 
\big(\rho^{-1}\big)_T=-u_X. 
\eeq 
Then we find that it is convenient to introduce a new dependent 
variable $\vartheta$, such that 
$$
m_1=\rho^2 \sin\vartheta, \qquad m_2 = \rho^2 \cos\vartheta, 
$$
and
we find that the system \eqref{2cpt} is transformed to a pair of equations for $\rho$ and $\vartheta$, namely (\ref{rhoinv}) and the 
relation 
\beq \label{vtheta}
\vartheta_T = v. 
\eeq 
However, the way that the relations \eqref{uveqs} between $u,v$ and the $m_j$ transform under the reciprocal transformation must also be taken into account. 

\begin{proposition} \label{recipsys} 
In terms of $\vartheta$ and another variable $\uppsi=\log\rho$, under the reciprocal transformation 
\eqref{RT} the system \eqref{2cpt} is transformed 
into the pair of equations 
\beq \label{altRTsys}
\begin{array}{rcl}
D_X \big( e^\uppsi \uppsi_{XT}-2e^{2\uppsi} \sin\vartheta \big)  + D_T \big( e^{-\uppsi}\big) &=&0, \\
\vartheta_{XT} + 2e^\uppsi \cos\vartheta & = & 0. 
\end{array}
\eeq 
The latter system is the compatibility condition for the Lax pair 
\beq\label{FGhatlin}
\hat{\Psi}_X = \hat{\bF}\hat{\Psi}, \qquad \hat{\Psi}_T=\hat{\bG}\hat{\Psi}, 
\eeq 
where 
$$
\hat{\bF} = \left(\begin{array}{cccc} 
0 & 0 & e^{-\uppsi} & 0 \\ 
0 & 0 & 0 & e^{-\uppsi} \\ 
\frac{1}{4}e^{-\uppsi} & -\ri\la e^{\uppsi+\ri\vartheta} & 0 & 0 \\ 
\ri\la e^{\uppsi-\ri\vartheta} & \frac{1}{4}e^{-\uppsi} & 0 & 0
\end{array}
\right) ,
$$
$$
\hat{\bG} = \left(\begin{array}{cccc} 
-\tfrac{1}{2}(\uppsi_T-\ri\vartheta_T) & 0 & 0 & \la^{-1} \\ 
0 & -\tfrac{1}{2}(\uppsi_T+\ri\vartheta_T) & \la^{-1} 
& 0 \\ 
 -\frac{1}{2}e^\uppsi \uppsi_{XT}+e^{2\uppsi} \sin\vartheta   & \frac{1}{4}\la^{-1} & \tfrac{1}{2}(\uppsi_T+\ri\vartheta_T) & 0 \\ 
\frac{1}{4}\la^{-1} & -\frac{1}{2}e^\uppsi \uppsi_{XT}+e^{2\uppsi} \sin\vartheta  & 0 & \tfrac{1}{2}(\uppsi_T-\ri\vartheta_T) 
\end{array}
\right) 
$$
\end{proposition}
\begin{proof} 
Under the reciprocal transformation, the relation between $u$ and $m_1$ produces 
$ (\rho D_X)^2 u - u = \rho^2 \sin\vartheta$, and this can be solved for $u$, with $u_X$ and $u_{XX}$ replaced in terms of $\rho$ from \eqref{rhoinv}, to obtain 
\beq\label{urhorel}
u = \rho(\log\rho)_{XT}-2\rho^2\sin\vartheta 
.
\eeq
Similarly, the second equation in \eqref{uveqs} gives the relation between $v$ and $m_2$, and this transforms to $\rho v_X=-2\rho^2\cos\vartheta$, in which $v$ can be replaced using   \eqref{vtheta}. After substituting the formula \eqref{urhorel} for $u$ back into \eqref{rhoinv}, 
this yields 
a closed system of two equations for 
$\rho$ and $\vartheta$, namely 
\beq \label{RTsys}
\begin{array}{rcl}
D_X \big( \rho (\log \rho)_{XT}-2\rho^2 \sin\vartheta \big)  + D_T \big( \rho^{-1}\big) &=&0, \\
\vartheta_{XT} + 2\rho\cos\vartheta & = & 0. 
\end{array}
\eeq 
After introducing $\uppsi$ via $\rho=e^\uppsi$, 
the above system can be rewritten in the form 
\eqref{altRTsys}. The same conclusion can be reached 
in a straightforward manner,  
by applying the reciprocal transformation \eqref{RT} directly to the linear system \eqref{FGlin}: indeed, if we specify the wave function as a function of the new independent variables, writing $\hat\Psi(X,T)=\Psi(x,t)$, 
then the transformed linear equations 
$$ 
\rho\hat\Psi_x = \bF\hat\Psi, \qquad 
\hat\Psi_T +\rho u\hat\Psi_X = \bG\hat\Psi
$$
immediately imply that 
$\hat \bF = \rho^{-1} \bF$, $\hat \bG= \bG - u\bF$. After substituting $m_2+\ri m_1 = \rho^2 e^{\ri\vartheta}$ into the original matrices \eqref{FGforms}, the compatibility conditions obtained from the new zero curvature equation 
$[D_X-\hat\bF,D_T-\hat\bG]=0$ consist of \eqref{vtheta} 
and \eqref{urhorel}, which can be rewritten 
in terms of $\uppsi$ as 
\beq\label{uformXT}
u = e^\uppsi \uppsi_{XT}-2e^{2\uppsi} \sin\vartheta, 
\eeq 
together with the system \eqref{RTsys}, or equivalently 
\eqref{altRTsys}. 
In the transformed Lax pair, we can then use  
\eqref{vtheta} 
and \eqref{urhorel} to write all the entries purely in terms of 
$\uppsi,\vartheta$ and their derivatives. 
\end{proof}

\begin{remark}
Note that the first equation of the system (\ref{altRTsys}) 
is written in conservation form, being equivalent to the equation \eqref{rhoinv}. This system also has the conservation law 
\beq\label{consRT}
\hat{\cal U}_T +  (2e^\uppsi \sin\vartheta)_X=0,  
\eeq 
where 
\beq\label{kdvvar}
\hat{\cal U} = -\frac{1}{2}\uppsi_{XX}-\frac{1}{4}\uppsi_X^2-\frac{1}{4}e^{-2\uppsi}
+\frac{1}{4}\vartheta_X^2.
\eeq
When $\theta=\pi/2$, the quantity $\hat{\cal U}$ becomes the 
standard KdV dependent variable, and  (\ref{altRTsys}) reduces to the first negative flow of the KdV hierarchy 
(see e.g.\ \cite{honech}). In the original coordinates, this reduction corresponds precisely to the special case $v=0=m_2$, when the two-component system reduces to the Camassa--Holm equation.  
\end{remark}

\subsection{Travelling wave reduction} \label{sec:travelling wave}

In this subsection, our main goal will be to describe the properties of the travelling wave solutions in the case $d=2$, namely for the two-component system \eqref{2cpt}. However, we begin by making some remarks about the case general of general $d$.  
We choose the travelling wave 
coordinate to be $z=x-ct$, let prime  denote $\frac{\rd}{\rd z}$ 
and assume zero background for the fields, which depend on $x$ and $t$ through the combination $z$ only, so that  
$$
u,\,\bm  \to 0 \qquad \mathrm{as} \qquad |z|\to\infty.
$$
Then from  \eqref{vecsys}
we obtain 
\begin{equation}\label{eq:M-ode}
-(c+u)\,\bm ' = 2u'\,\bm + V\bm.  
\end{equation}
Taking the inner product with $\bm$, and using the fact that the matrix $V$ is skew-symmetric, we obtain

\begin{equation}\label{eq:E-inv}
\frac{\rd}{\rd z}\log\big(|\bm|^2)\big) = -\frac{4u'}{c+u},
\end{equation}
and hence find  the invariant
\begin{equation}\label{eq:E-curve}
K= |\bm |^2 \,(c+u)^4 = \mathrm{const}.
\end{equation}
Zero background gives $K=0$, which by \eqref{constraints} forces $(u''-4\alpha u,\bv')\equiv(0,0)$ for smooth profiles; therefore 
nontrivial travelling waves with zero background cannot be smooth, so must be peaked/cusped, being realized when $c+u$ vanishes at a crest. So we have 
\begin{proposition}\label{peakedwave} For any $d$, the system  \eqref{vecsys} has no nontrivial smooth travelling waves with zero background. 
\end{proposition}
%
In the light of the above result, the latter  sections of the paper  will focus on peaked soliton (peakon) solutions of \eqref{vecsys},  which are weak solutions. The rest of this subsection is 
concerned with the smooth travelling waves for $d=2$. 

In order to analyze the ODEs for travelling waves of the two-component system \eqref{2cpt}, we will make use of the reciprocal transformation \eqref{RT}, which can be reduced to the level of these ODEs. Our treatment relies on a general feature of similarity reductions of reciprocal transformations (see \cite{bh}, for instance), which is that under a change of independent variables, the roles of parameters and constants of motion can be switched. In particular, in the case of hodograph transformations (changing the time variable) in Hamiltonian systems, this phenomenon is referred to as ``coupling constant metamorphosis'' \cite{meta1, meta2}. For the context at hand, we can state the following

\begin{lemma} \label{lem:RTV}
Under the reciprocal transformation given by 
\beq\label{rfrt}
\rd X = \rho\, \rd x +{\cal F}\, \rd t, \qquad 
\rd T = \rd t , 
\eeq 
the travelling wave solutions of the conservation law 
$$ 
\rho_t = {\cal F}_x, 
$$
moving with speed $c$, 
admit the conserved quantity 
\beq\label{kcons}
k = {\cal F} + c\rho =\mathrm{const}.
\eeq 
Moreover, for a fixed value of this constant, these solutions correspond to travelling waves of the reciprocally transformed PDE 
\beq\label{rtpde}
(\rho^{-1})_T + ({\cal F} \rho^{-1})_X = 0
\eeq 
moving with speed $k$. 
\end{lemma}
\begin{proof}
Upon reducing to travelling waves with speed $c$, 
the conservation law with density $\rho$ and flux 
$\cal F$ produces the ODE $-c\rho' = \cal F'$, where the prime denotes $\frac{\rd}{\rd z}$ as before. This can be 
rearranged as 
$$ 
\frac{\rd}{\rd z} \left( {\cal F} + c\rho \right) =0, 
$$
which integrates to give the conserved quantity $k$ as in 
\eqref{kcons}. 
Now consider the reduction of the reciprocal transformation \eqref{rfrt} to the travelling wave solutions: by a slight abuse of notation, we have the density $\rho = \rho(z)$, being a function of the similarity variable $z$, and  
$$ 
\rho \,\rd z = \rho (\rd x-c\,\rd t) 
= \rd X-{\cal F}  \, \rd t -c \rho\,\rd t = \rd X -k\,\rd T. 
$$
This can be rewritten as 
\beq\label{bigsz}
\rd Z = \rho(z)\, \rd z, 
\eeq 
where $Z=X-k T$ is a travelling wave variable in the new $X,T$ coordinates, with wave speed $k$. Then, with this value of the wave speed, we see from \eqref{rtpde}  that the travelling waves of the reciprocally transformed PDE satisfy 
$$ 
\frac{\rd}{\rd Z} \left(  k\rho^{-1} - {\cal F}\rho^{-1} \right) =0, 
$$
which integrates and rearranges to yield \eqref{kcons}, where now $c$ is an integration constant. Thus we see that, subject to the relation \eqref{bigsz} between the independent variables, there is a one-to-one correspondence between travelling wave solutions of the two reciprocally related PDEs, modulo the metamorphosis of coupling constants $c\leftrightarrow k$, exchanging the roles of these parameters. 
\end{proof}

The above result is directly applicable to the case at hand, where we have the reciprocal transformation  \eqref{RT}, with the flux of the conservation law being ${\cal F} = \rho u$. In the light of the preceding 
lemma, to determine the travelling waves of the system 
\eqref{2cpt} it will be sufficient to consider travelling waves of \eqref{altRTsys} moving with speed $k\neq 0$, where the original wave speed $c$ will appear as a parameter. Thus we take 
\beq\label{realvars} \uppsi = \psi(Z), \qquad \vartheta = \theta (Z), 
\qquad Z = X-kT .
\eeq 
In the course of doing so,  for the Hamiltonian theory it will be 
convenient to apply imaginary rescalings to some of the variables/parameters, and consider solutions of 
the form 
\beq\label{trans} \uppsi = \psi(Z), \qquad \vartheta = \ri \varphi(Z), \qquad Z = X-\ri CT, \qquad \mathrm{where} \quad k=\ri C, 
\eeq 
with the arbitrary constant $C\neq 0$ (excluding the stationary case $C=0$). 
In that case, the PDEs \eqref{altRTsys} reduce to a pair of ODEs for $\psi,\varphi$, given by 
\beq \label{ODEsys}
\begin{array}{rcl}
\frac{\rd}{\rd Z} \big( C e^\psi \psi''+2e^{2\psi} \sinh\varphi+C e^{-\psi}\big) &=&0, \\
C\varphi'' + 2e^\psi \cosh\varphi & = & 0, 
\end{array}
\eeq 
where now the prime denotes $\frac{\rd}{\rd Z}$. 
Since the first equation above is a total derivative, we can integrate it once, to obtain the constant of motion 
\beq\label{ddef}
K=-e^\psi \psi''-\frac{2}{C}e^{2\psi} \sinh\varphi-e^{-\psi}, \qquad \mathrm{where} \quad K=\frac{\ri c}{C} 
\eeq 
by the result of Lemma \ref{lem:RTV}. Note that we will ultimately be interested in the case where  $c$, the speed of the travelling waves of the two-component CH system, is real, and the parameter $C=-\ri k$ is pure imaginary, so the constant $K$ will also be real. 

After performing an integration as above, to obtain \eqref{ddef}, we find that overall the system \eqref{ODEsys} for the reciprocally transformed travelling wave reduction can be written in canonical Hamiltonian form, as 
\beq\label{hameq}
\begin{array}{rclrclrclrcl}
\psi'& =&\frac{\partial h}{\partial p_\psi}, &
\varphi'& =&\frac{\partial h}{\partial p_\varphi}, &
p_\psi'& =&-\frac{\partial h}{\partial \psi}, & 
p_\varphi'& =&-\frac{\partial h}{\partial \varphi}, 
\end{array}
\eeq 
where $p_\psi=\psi'$ and $p_\varphi=\varphi'$ are canonically conjugate momenta, and 
the Hamiltonian function is 
\beq\label{ham} h=
\tfrac{1}{2}p_\psi^2
+\tfrac{1}{2}p_\varphi^2
+2C^{-1}\,e^\psi\sinh\varphi-\tfrac{1}{2}e^{-2\psi}-Ke^{-\psi} .
\eeq 
In order to prove that the latter Hamiltonian system is completely integrable in the Liouville sense, we 
need to find a second independent constant of motion that is in involution with $h$, which 
is achieved by considering the reduction of the Lax pair to the travelling waves. 

\begin{theorem}\label{liouvint}
Via the hodograph transformation \eqref{bigsz} obtained from the density $\rho$ in \eqref{rhodefn2}, and the change of parameters and dependent variables \eqref{trans}, with $\rho =e^\psi$ and $c=-\ri \,CK$, the travelling wave solutions of the PDE system \eqref{2cpt} with speed $c$ correspond to the solutions of the Hamilton's equations \eqref{hameq} defined by \eqref{ham}. This Hamiltonian system with two degrees of freedom is the compatibility condition for the linear system   
\beq\label{LMlin}
\tilde{\Psi}' = \bM\tilde{\Psi}, \qquad \bL\tilde{\Psi}=\mu\tilde{\Psi}, 
\eeq 
where 
$$
\bL = 
\left(\begin{array}{cccc} 
\frac{1}{2}C(p_\psi+p_\varphi) & 0 & Ce^{-\psi} & -\ri\,\la^{-1} \\ 
0  &  \frac{1}{2}C(p_\psi-p_\varphi) & -\ri\,\la^{-1}  & Ce^{-\psi}  \\ 
-\frac{1}{4}C(e^{-\psi}+2K) 
& -\frac{\ri}{4}\la^{-1} -\ri\, Ce^{\psi-\varphi} \la & -\frac{1}{2}C(p_\psi-p_\varphi)     & 0 \\ 
-\frac{\ri}{4}\la^{-1} +\ri \,Ce^{\psi+\varphi} \la 
 & 
-\frac{1}{4}C(e^{-\psi}+2K) 
& 0 & \frac{1}{2}C(p_\psi+p_\varphi) 
\end{array}
\right) 
$$
and
$$
\bM = \left(\begin{array}{cccc} 
0 & 0 & e^{-\psi} & 0 \\ 
0 & 0 & 0 & e^{-\psi} \\ 
\frac{1}{4}e^{-\psi}
  & -\ri e^{\psi-\varphi}\la  & 0 & 0 \\ \ri e^{\psi+\varphi}\la
 & \frac{1}{4}e^{-\psi} & 0 & 0 \end{array}
\right) .
$$
Moreover, it is completely integrable in the Liouville sense, having a second independent first integral, given by 
\beq\label{Jham}
J = p_\psi^2p_\varphi^2 
+4C^{-1}\,e^{\psi}\cosh \varphi \, p_\psi p_\varphi 
-e^{-\psi}\big(e^{-\psi}+2K\big)\,  p_\varphi^2
+4C^{-2}\,e^{2\psi}\cosh^2 \varphi 
-4KC^{-1}\,\sinh \varphi .
\eeq 
\end{theorem}

\begin{proof}
The travelling wave reduction of the Lax pair \eqref{FGhatlin} is obtained by 
taking the wave function in the form $\hat{\Psi}=\exp(\ri\mu T) \tilde{\Psi}(Z)$, which leads to the pair of 
linear equations 
(\ref{LMlin}). 
The compatibility condition for the latter linear  system  is the matrix Lax equation 
$$ 
\bL'=[\bM,\bL], 
$$
which is equivalent to Hamilton's equations \eqref{hameq} for $\psi,\varphi$ and their conjugate momenta. 
Moreover, the spectral curve 
\beq\label{spec}
{\cal P}(\la^2,\mu^2) =\la^4 \det \big(\bL(\la)-\mu \mathbf{1}\big)=0
\eeq 
is invariant under the evolution defined by the Lax equation, 
being given by the polynomial 
$$ 
{\cal P} = \frac{1}{16}+\gamma_1\,\la^2 +\gamma_2\,\la^4 -C^4\la^6 
+\left(\frac{1}{2}\la^2 -\gamma_3\,\la^4\right)\mu^2 
+\la^4\mu^4, 
$$
where each of the coefficients $\gamma_1,\gamma_2,\gamma_3$ is a first integral of the system. After further calculations with computer algebra (MAPLE) we find that these coefficients are given explicitly by 
$$ 
\gamma_1=  \frac{1}{4}C^2(K^2-h), 
 \qquad \gamma_2 = \frac{1}{4}C^4(h^2-J), 
\qquad \gamma_3 = C^2 h, 
$$
where the function 
\eqref{Jham}
is in involution with $h$, that is 
$\{h,J\}=0 $ with respect to the canonical Poisson bracket, so $J$ is a second independent constant of motion, implying Liouville 
integrability of the system. 
\end{proof} 
\begin{remark}
If we let $\cal C$ denote the curve defined by the affine equation 
${\cal P}(\la^2,\mu^2)=0$ in the $(\la,\mu)$ plane, then this is an unramified double cover of the curve ${\cal C}'$ defined by  ${\cal P}(\ze,\mu^2)=0$ in the $(\ze,\mu)$ plane, which is in turn a ramified 
double cover of the elliptic curve $\cal E$ defined by 
${\cal P}(\ze,\eta)=0$, that is 
$$ 
{\cal E}: \qquad 
\ze^2\eta^2 -(\gamma_3\,\zeta -\tfrac{1}{2})\, \ze\eta
-C^4\ze^3+\gamma_2\,\ze^2+\gamma_1\,\ze+\tfrac{1}{16}=0, 
$$
so that we have a sequence of coverings: 
$$
{\cal C}\overset{2:1}{\longrightarrow}{\cal C}'
\overset{2:1}{\longrightarrow}{\cal E}.
$$
By applying the Riemann-Hurwitz
theorem to each covering in succession, we see that the genus of 
${\cal C}'$ is 3, being a double cover of the genus 1 curve 
$\cal E$ with 4 ramification points, while the genus of $\cal C$ is 5. General theory tells us that the flows of the Hamiltonian system 
linearize on  $\Jac {\cal C}$, the Jacobian variety of the spectral curve for the Lax matrix $\bL$. However, the latter is a 5-dimensional abelian variety, whereas the invariant tori defined by the level 
sets $h=\mathrm{const}$,  $J=\mathrm{const}$ are two-dimensional in this case (cf.\ Fig. \ref{figb}); so ideally we would like to construct a separation of variables that describes straight line motion on a two-dimensional 
subvariety of $\Jac {\cal C}$. There are various results on separating coordinates for natural Hamiltonians with two degrees of freedom having a second invariant that is quartic in momenta \cite{rrg}, including the stationary flow of the Hirota-Satsuma system \cite{BEF}. 
\end{remark}
\begin{figure}
\centering
\includegraphics[height=2.5in, width=4.2in] 
{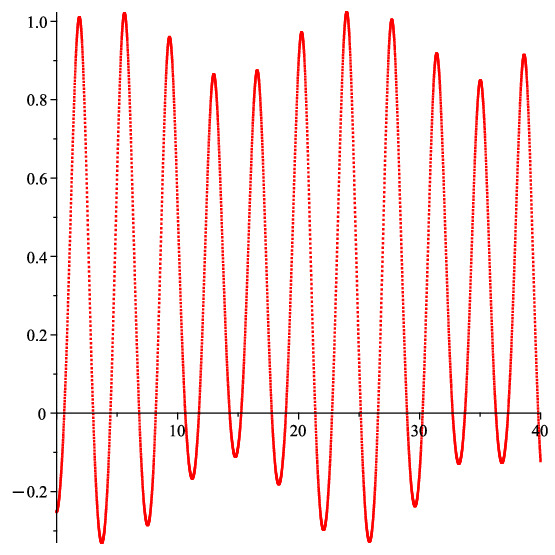}
\caption{Travelling wave profile of $\hat{\cal U}$ given by \eqref{kdvvar}, plotted against $Z$, 
obtained from the solution of the system \eqref{realpossys} with parameters $k=1$, $K=-3$ and the initial conditions taken as $\psi(0)=-0.3$, $\theta(0)=1.520796327,\psi'(0)=-0.05,\theta'(0)=0.15$. 
}
\label{travwaveUt}
\end{figure}

The canonical coordinates and momenta $\psi,\varphi,p_\psi,p_\varphi$ and the parameter $C$ are convenient for interpreting the travelling wave reduction as a natural Hamiltonian system with two degrees of freedom, and showing that the system is Liouville integrable, as above. However, for \textit{real} travelling waves of 
\eqref{altRTsys}  
moving with real speed $k=\ri C\neq 0$, 
and corresponding real travelling waves of \eqref{2cpt} with speed $c=-kK$, 
it is preferable to use the coordinates $(\psi,\theta, \psi',\theta')\in\R^4$, where $\theta = \ri \varphi$ and the subscripts denote derivatives with respect to the travelling wave variable $z=X-kT$. In those coordinates, the equations of motion 
\beq \label{realpossys}
\frac{\rd}{\rd Z}
\left(\begin{array}{c} \psi \\ \theta \\ \psi' \\
 \theta' \end{array}\right) = 
 \left(\begin{array}{c}
\psi' \\  \theta' \\-2k^{-1}e^{\psi}\sin\theta -Ke^{-\psi} -e^{-2\psi}\\  2k^{-1}e^{\psi}\cos\theta
\end{array}\right)
\eeq 
have a fixed point when 
\beq\label{fp1}
\psi = \log Q, \quad  \theta = \pm \frac{\pi}{2}, \quad \psi'=0, \quad \theta'=0, 
\eeq 
where $Q$ is a root of the cubic 
\beq\label{fp2}
\pm 2k^{-1}Q^3+KQ+1=0. 
\eeq 
Then it is instructive to consider solutions of the 4D system close to an elliptic fixed point, which should correspond to compact 2D Liouville tori. In particular, when $k=1$, $K=-3$ (corresponding to waves in \eqref{2cpt} with speed $c=3$) 
and the plus sign is chosen in \eqref{fp1} and \eqref{fp2}, the cubic has a real root at $Q=1$, and  at the fixed point with $\psi=0$, $\theta=\pi/2$ the Jacobian matrix of the 4D system \eqref{realpossys} has eigenvalues $\pm\sqrt{2}\ri,\pm\sqrt{3}\ri$. A particular travelling wave solution, showing the profile of the potential $\hat{\cal U}$ as in \eqref{kdvvar}, the profile of  $u$ against $z(Z)$,  and a 3D projection of the corresponding Liouville torus, is plotted in Figs. \ref{travwaveUt}, \ref{figa} and \ref{figb}, respectively. In these plots, the initial conditions are 
$ \psi(0)=-0.3$, $\theta(0)=\frac{\pi}{2} -0.05$, 
$\psi'(0)=-0.05$, $\theta'(0)=0.15$ and $z(0)=0$, with the parameter values $k=1$, $K=-3$, so $c=3$ is the speed of the wave in the system \eqref{2cpt}. Note that \eqref{bigsz} gives  $z'=\rho^{-1}=e^{-\psi(Z)}$, so we have fixed 
$z(Z)=\int_0^Z e^{-\psi(S)}\rd S$ to make the origin of $z$ and $Z$ coincide. Also, from \eqref{kcons} with ${\cal F}=\rho u$ in this case, we have 
$$ u =k\rho^{-1}-c, 
$$
which gives $u=e^{-\psi}-3$ in this particular numerical example, allowing the parametric curve 
$\big(z(Z),u(Z)\big)$ to be plotted in Fig.\ref{figa}. 

\begin{figure}[ht]
\begin{minipage}[b]{.45\textwidth}
\centering
\includegraphics[width=1\textwidth]{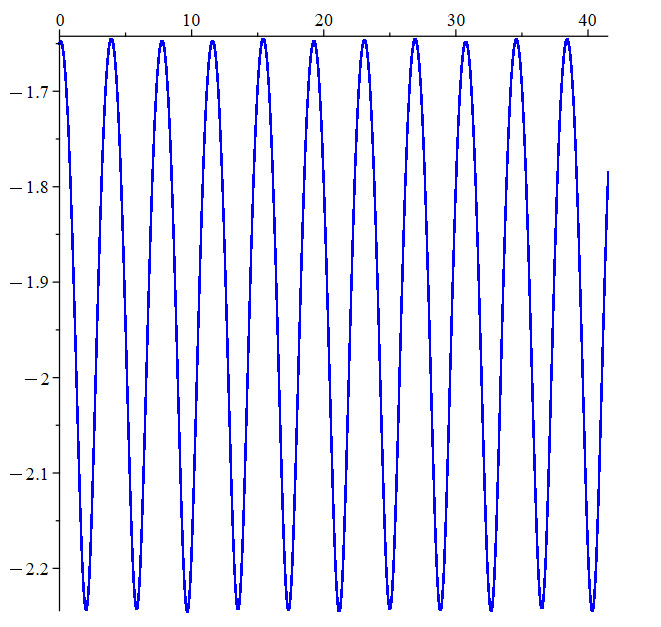}
\caption{ Travelling wave for \eqref{2cpt}, plotting $u$ parametrically vs.\ $z=x-ct$, with speed $c=3$, obtained from the solution of \eqref{realpossys} with initial data as in Fig.\ref{travwaveUt}.}
\label{figa}
\end{minipage}
\hfill
\begin{minipage}[b]{.5\textwidth}
\centering
\includegraphics[width=1\textwidth]{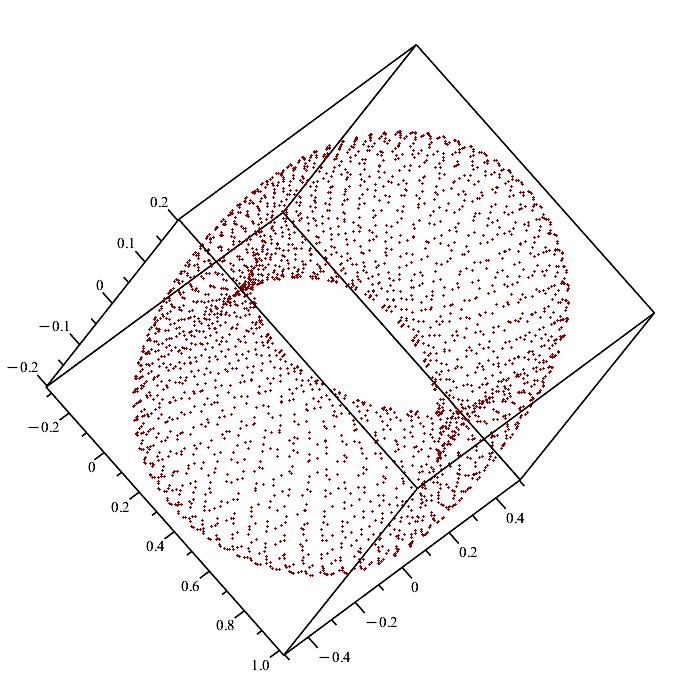}
\caption{Liouville torus for the 4D system \eqref{realpossys}, projected onto the 3D space with coordinates 
  $(\hat{\cal U},\psi',\theta')$, obtained from the same initial data as in Fig.\ref{travwaveUt}.}
  \label{figb}
\end{minipage}
\end{figure}

\subsection{Connection with the Hirota--Satsuma system} \label{sec:Hirota-Satsuma}

As mentioned above, the vector Camassa--Holm system 
(\ref{vecsys}) has a third order symmetry, given by 
\eqref{3rdorderflow}. In the two-component case, when $d=2$, this symmetry is given by 
\beq\label{3rdnis2}
\begin{array}{rcl}
m_{1,\tau} & = & \tfrac{1}{2}(4\al D_x-D_x^3) \left(\frac{m_1}{(m_1^2+m_2^2)^\frac{3}{4}}\right) , 
\\
m_{2,\tau} & = & \tfrac{1}{2}(4\al D_x-D_x^3) \left(\frac{m_2}{(m_1^2+m_2^2)^\frac{3}{4}}\right) .
\end{array}
\eeq
(We have restored the arbitrary parameter $\al$ here.) 
The reciprocal transformation \eqref{RT} can be extended to this 
symmetry, 
and indeed, to the whole hierarchy of higher symmetries. 
In order to do this, one should 
take $$\rd X=\rho\, \rd x + \rho u\, \rd t +\si \rd \tau, \qquad  \rd T = \rd t, \qquad 
\rd \tau'=\rd \tau, 
$$  
where $\si$ is the flux of the density $\rho$ with respect to the 
$\tau$ flow, and $\tau'$ is the new time variable corresponding to 
the transformation of $\tau$ under the extension of \eqref{RT}.  We are then able to relate the transformed symmetry to a well-known coupled third order  system of PDEs, namely the Hirota--Satsuma system \cite{hirotasatsuma, hirotasatsuma2}
\beq\label{hs} 
\begin{array}{rcl}
{\cal U}_{\tau'} & = & \tfrac{1}{2}{\cal U}_{XXX}+3 \,{\cal U}{\cal U}_X - 6{\cal V}{\cal V}_X, \\ 
{\cal V}_{\tau'} & = & -{\cal V}_{XXX}-3 \,{\cal U}{\cal V}_X .
\end{array}
\eeq 
This is described by the following result, which is verified by a direct calculation. 
\begin{proposition}\label{hsmiura}
The extension of the reciprocal    transformation 
\eqref{RT} to the third order symmetry \eqref{3rdnis2} 
produces the coupled system 
\beq\label{3rdrhoth}
\begin{array}{rcl}
\uppsi_{\tau'} & = & \,\,\tfrac{1}{2}\uppsi_{XXX}-\tfrac{1}{4}\uppsi_X^3 
-\tfrac{3}{4}\vartheta_X^2\uppsi_X +\tfrac{3}{2}\vartheta_X\vartheta_{XX} 
-3\al e^{-2\uppsi}\uppsi_X, 
\\ 
\vartheta_{\tau'} & = & -\vartheta_{XXX}
+\tfrac{3}{2}\vartheta_X\uppsi_{XX}
+\tfrac{3}{4}\vartheta_X\uppsi_X^2 
+\tfrac{1}{4}\vartheta_X^3 
+3\al e^{-2\uppsi}\vartheta_X, 
\end{array}
\eeq 
which is related to the Hirota--Satsuma system (\ref{hs}) via the 
Miura map 
\beq\label{hsmiuramap}
{\cal U} =  \frac{1}{2}\uppsi_{XX}-\frac{1}{4}\uppsi_X^2
-\al e^{-2\uppsi}+2\sqrt{\al}e^{-\uppsi}\uppsi_X+\frac{1}{4}\vartheta_X^2, 
\qquad 
{\cal V} = \frac{\ri}{2}\big(\vartheta_{XX}-\vartheta_{X}\uppsi_X+2\sqrt{\al}e^{-\uppsi}\vartheta_{X}\big)
.
\eeq 
\end{proposition}
    
\begin{remark}
The system \eqref{3rdrhoth} admits another Miura map whose 
first component $\tilde{\cal U}$ looks very similar to ${\cal U}$,  being given by 
\beq\label{2ndmiura}
\tilde{\cal U} = -\frac{1}{2}\uppsi_{XX}-\frac{1}{4}\uppsi_X^2
-\al e^{-2\uppsi}-\frac{1}{4}\vartheta_X^2, \qquad 
\tilde{\cal V} =\vartheta_X^2 .
\eeq
The quantity $\tilde{\cal U}$ is related to the density  for the 
conservation law \eqref{consRT} by taking 
$\tilde{\cal U}=\hat{\cal U}-\tfrac{1}{2}\vartheta_X^2$ (where previously $\hat{\cal U}$ in \eqref{kdvvar} was written in the 
specific case $\al=\tfrac{1}{4}$). Upon applying the second Miura map  (\ref{2ndmiura}), we find a third order flow for  
 $\tilde{\cal U}$ and  $\tilde{\cal V}$ that contains powers of 
  $\tilde{\cal V}$ as denominators on the right-hand side, that is  
\beq\label{miura2sys}
\begin{array}{rclrl} 
\tilde{\cal U}_{\tau'} & =& \,\tfrac{1}{2}\,\tilde{\cal U}_{XXX}+\,3\,\tilde{\cal U}\tilde{\cal U}_X 
&+&\tfrac{3}{2}\tilde{\cal V}\tilde{\cal U}_X
-\tfrac{3}{8}\tilde{\cal V}^{-1}\tilde{\cal V}_X\tilde{\cal V}_{XX}
+\tfrac{3}{16}\tilde{\cal V}^{-2}\tilde{\cal V}_X^3
,
\\ 
\tilde{\cal V}_{\tau'} & =& 
-\tilde{\cal V}_{XXX}-\tfrac{3}{2}\tilde{\cal V}\tilde{\cal V}_X 
-3\tilde{\cal V}_X\tilde{\cal U}
&-&6\tilde{\cal V}\tilde{\cal U}_X
\,\,+\tfrac{3}{2}\tilde{\cal V}^{-1}\tilde{\cal V}_X\tilde{\cal V}_{XX}
\,\,-\tfrac{3}{4}\tilde{\cal V}^{-2}\tilde{\cal V}_X^3
. 
\end{array}
\eeq 
\end{remark}

\section{Two-component deformations of the Camassa--Holm equation} \label{sec:classification}

In this section we classify integrable deformations of the Camassa--Holm equation of the form

\beq \label{classify}
\begin{array}{rclcl}
(1-D_x^2)u_t & = &3uu_x-2u_xu_{xx}-uu_{xxx}+f(u,v,u_x,v_,u_{xx},v_{xx},u_{xxx},v_{xxx}) \\
v_{xt} & = & g(u,v,u_x,v_x,u_{xx},v_{xx},u_{xxx},v_{xxx}).
\end{array} 
\eeq 

In Theorem \ref{class2CH} below, we list all integrable systems (\ref{classify}) modulo rescaling transformations $u\to\alpha u,\,v\to\beta v,\,x\to\gamma x,\,t\to \delta t$ satisfying the following conditions:
 \begin{itemize}
 \item[I.] Functions $f,g$ are homogeneous quadratic polynomials in $u,v$ and their derivatives up to order 3;
 \item[II.] System possesses a reduction $v=0$;
 \item[III.] System (\ref{classify}) possesses a local symmetry of order 5, which upon imposing the reduction $v=0$ becomes of order $3$;
 \item[IV.] System (\ref{classify}) is non-triangular;
\item[V.] Upon shifting $u\to u+1$ the system gains a diagonal non-degenerate linear term.
 \end{itemize}

Condition II implies that upon reduction $v=0$ the resulting system is the scalar Camassa--Holm equation, so the system (\ref{classify}) is indeed a deformation.

As for condition III, we assume here as the definition of integrability the existence of an infinite hierarchy of local higher symmetries. We seek local symmetries in the form 
 \beq\label{sym}
 \begin{array}{rcl}
 u_\tau & = & F, \\ 
 v_\tau & = & G ,  
 \end{array}
 \eeq 
 where $F$ and $G$ are required to be local functions 
 of $u,v$, i.e. functions of $u,v$ and their derivatives up to a finite order. By definition system (\ref{sym}) is a generator of a symmetry of system (\ref{classify}) if
\begin{eqnarray*}
(1-D_x^2)D_t(F)&=&D_{\tau}\left(3uu_x-2u_xu_{xx}-uu_{xxx}+f(u,v,u_x,v_,u_{xx},v_{xx},u_{xxx},v_{xxx})\right),\\
D_xD_t(G)&=&D_{\tau}\left(g(u,v,u_x,v_x,u_{xx},v_{xx},u_{xxx},v_{xxx})\right),
\end{eqnarray*}
where $D_t,D_{\tau}$ are evolutionary derivations defined by (\ref{classify}) and (\ref{sym}). 
The maximal order of derivatives of $u$ or $v$ on which functions $F$ or $G$ depend is called the order of a symmetry. The scalar Camassa--Holm equation possesses local higher symmetries of orders congruent to $1\bmod 2$. For the class of systems being considered here, the existence of a symmetry of order 5 is necessary for the existence of infinitely many symmetries.

We shall call the system (\ref{classify}) triangular if 
$$\frac{\partial f}{\partial v}=\frac{\partial f}{\partial v_x}=\frac{\partial f}{\partial v_{xx}}=\frac{\partial f}{\partial v_{xxx}}=0$$ 
or 
$$\frac{\partial g}{\partial v}=\frac{\partial g}{\partial v_x}=\frac{\partial g}{\partial v_{xx}}=\frac{\partial g}{\partial v_{xxx}}=0,$$
i.e. either of the equations depends only on its own variable. According to condition IV, we shall disregard triangular systems in the classification.

Condition V is most technical of the five requirements above. It means that, upon applying the shift $u \to u+1$, the system (\ref{classify}) becomes 
$$ 
D_t \, \left(\begin{array}{c} (1-D_x^2)u \\ v_x \end{array}\right) = 
\left(\begin{array}{cc} 3D_x -D_x^3 & 0 \\ 0 & {\cal L} \end{array} \right) \left(\begin{array}{c} u \\ v \end{array}\right) 
+ \left(\begin{array}{c} 3uu_x-2u_xu_{xx}-uu_{xxx}+f \\ g \end{array}\right), $$
where ${\cal L}$ is a linear differential operator with constant coefficients. This condition is required for the explicit derivation of integrability conditions in the framework of the symmetry approach, as described in \cite{HNW}, for instance. Here we omit all further  technical details 
of deriving these integrability conditions, which require extensive use of computer algebra, and proceed to state the final result of the derivation. 

The following classification theorem holds:

\begin{theorem} \label{class2CH}
Every system satisfying conditions I-V is equivalent modulo $u\to\alpha_1 u,\,v\to\alpha_2 v,\,x\to\alpha_3 x,\,t\to \alpha_4 t,\,\,\alpha_1,\ldots,\alpha_4\in\C^*$ to one of the list:
\begin{eqnarray}
\label{jc1u}
(1-D_x^2)u_t&=&3uu_x-2u_xu_{xx}-uu_{xxx}+\frac{1}{2}D_x(v^2),\\ \nonumber
v_{xt}&=&-(1-D_x^2)(uv),
\\ \nonumber \\
\label{jc2u}
(1-D_x^2)u_t&=&3uu_x-2u_xu_{xx}-uu_{xxx}+\frac{1}{2}D_x\left((1+D_x)v\right)^2,\\ \nonumber
v_{xt}&=&-(1-D_x)\left(u(1+D_x)v\right),\\ \nonumber \\
\label{mjc1u}
(1-D_x^2)u_t&=&3uu_x-2u_xu_{xx}-uu_{xxx}+\frac{1}{2}D_x\left((1-D_x^2)v\right)^2,\\ \nonumber
v_{xt}&=&-u(1-D_x^2)v,
\\ \nonumber \\
\label{mjc2u}
(1-D_x^2)u_t&=&3uu_x-2u_xu_{xx}-uu_{xxx}+\frac{1}{2}D_x(2vv_{xx}+v_x^2-v^2),\\ \nonumber
v_{xt}&=&vu_{xx}+u_xv_x+uv_{xx}-uv+2vv_{xx}+v_x^2,
\\ \nonumber \\
\label{CLZ1u}
(1-D_x^2)u_t&=&3uu_x-2u_xu_{xx}-uu_{xxx}+\frac{1}{2}D_x(v^2),\\ \nonumber
v_{xt}&=&(uv)_{xx},\\ \nonumber \\
\label{CLZ2u}
(1-D_x^2)u_t&=&3uu_x-2u_xu_{xx}-uu_{xxx}+\frac{1}{2}D_x(v_x^2),\\ \nonumber
v_{xt}&=&(uv_x)_{x}.
\end{eqnarray}
\end{theorem} 

\begin{remark} It is worthwhile to comment in turn  on each item in the above list. 
\begin{itemize}
\item The first system, namely \eqref{jc1u}, is equivalent 
(up to rescaling) to the $d=2$ case of the vector Camassa--Holm system \eqref{vecsys}.

\item System (\ref{jc2u}) is related  to 
(\ref{jc1u}) via the Miura transformation $v\to (1+D_x)v$.

\item System (\ref{mjc1u}) 
is related  to (\ref{jc1u}) via the
Miura transformation $v\to (1-D_x^2)v$. 

\item System (\ref{CLZ1u}) is the 2CH system studied in \cite{CLZ,falqui,LZ}  
(see also equation (43) in \cite{or}). 

\item System (\ref{CLZ2u}) is the potential version of the latter one, being related to  (\ref{CLZ1u}) via 
$v\to v_x$. 

\item The system \eqref{mjc2u} appears to be new.
\end{itemize} 
 
\end{remark}

Upon rescaling $x\to2\alpha^{\frac{1}{2}}x,\,\,u\to2\alpha^{\frac{1}{2}}u$ and
introduction of the variables 
$$m_1=\frac{1}{2}u_{xx}-2\alpha u,\qquad m_{2}=-\frac{1}{2}v_x$$ the systems (\ref{jc1u})-(\ref{CLZ2u}) can be rewritten as

\begin{eqnarray}
 \label{jc1}
 m_{1,t}&=&2u_xm_1+um_{1,x}+vm_2,\\ \nonumber
 m_{2,t}&=&2u_xm_2+um_{2,x}-vm_1,\\ \nonumber\\ 
 \label{jc2}
 m_{1,t}&=&2u_xm_1+um_{1,x}+\frac{1}{2\alpha}(2\sqrt{\alpha}m_2+m_{2,x})(\sqrt{\alpha}v-m_2),\\ \nonumber
 m_{2,t}&=&(um_2)_x+\sqrt{\alpha}v(2\sqrt{\alpha} u-u_x),\\ \nonumber\\
 \label{mjc1}
 m_{1,t}&=&2u_xm_1+um_{1,x}+\frac{1}{8\alpha^2}(4\alpha m_2-m_{2,xx})(m_{2,x}+2\alpha v),\\ \nonumber
 m_{2,t}&=&u(m_{2,x}+2\alpha v),
 \\ \nonumber\\
 \label{mjc2}
 m_{1,t}&=&2u_xm_1+um_{1,x}-\alpha^{-1}m_2m_{2,x}-\frac{1}{4\alpha}v(4\alpha m_{2}-m_{2,xx}),\\ \nonumber
 m_{2,t}&=&(um_2)_x-m_1v+\alpha^{-\frac{1}{2}}(vm_{2,x}-m_2^2),
 \\ \nonumber\\
\label{CLZ1}
m_{1,t}&=&2u_xm_1+um_{1,x}+vm_2,\\ \nonumber
 m_{2,t}&=&2u_xm_2+um_{2,x}-vm_1-2\alpha uv,
 \\ \nonumber\\
\label{CLZ2}
m_{1,t}&=&2m_1u_x+um_{1,x}-\frac{1}{2\alpha}m_2m_{2,x},\\ \nonumber
 m_{2,t}&=&(um_2)_{x},
 \end{eqnarray}

With the variables $m_1,m_2$ as above, 
\eqref{jc1} is just the $d=2$ case of the vector Camassa--Holm system \eqref{vecsys}. The systems \eqref{jc2}
and  \eqref{mjc1} can also be extended to vector analogues, related to \eqref{vecsys} via Miura transformations. 

\begin{example}
In the limit $\alpha=0$, the system \eqref{jc1} becomes 
\begin{equation}\label{hs2a}
\begin{array}{rcl} 
u_{xxt} & = & (uu_{xx})_x + 
\frac{1}{2}(u_x^2)_x -\frac{1}{2} (v^2)_x, 
\\ 
v_{xt} & = & (uv)_{xx}, 
\end{array}
\end{equation}
which we refer to as the Hunter-Saxton limit with two components, so this is Hunter-Saxton-2, or HS2 for short.
Upon integrating both equations with vanishing boundary conditions at infinity, this gives 
\begin{equation}\label{hs2b}
\begin{array}{rcl} 
u_{xt} & = & uu_{xx} + 
\frac{1}{2} u_x^2 -\frac{1}{2} v^2, 
\\ 
v_{t} & = & (uv)_{x}. 
\end{array}
\end{equation} 
 In the remainder of the paper, we study the weak solutions for the multi-component systems of Hunter-Saxton type.  
\end{example}

\begin{remark} \label{rem:simanddif}

We want to comment on the similarities and differences between 
 \eqref{jc1} and \eqref{CLZ1}.  
 Our first observation is that \eqref{jc1} and \eqref{CLZ1} share the same $\alpha=0$ limit.  In other words, the Hunter-Saxton limits for these two 
equations coincide.  Yet the natural boundary conditions are very different in the two cases.  We will elaborate on this in the next section.   However, it is perhaps even more pressing to ask whether these two equations for $\alpha\neq 0$ are substantially different. We will address this question for the 
class of weak solutions, focusing on solutions $m_1, m_2$ with singular support only (peakons), consisting of sums of Dirac deltas supported on moving points 
$\{x_j(t)\}$.  In this case, $v$ is a primitive of $m_2$ and must be piecewise constant, while $u$ is a continuous, piecewise smooth function involving weighted sums of $e^{-2\abs{x-x_j}}$.  
It is then easy to check that the left-hand sides involve a distribution of order $1$ with singular support only (thus containing Dirac deltas and distributional derivatives of the Dirac deltas), so the right-hand sides in both equations must be of the same type.  
However, the last term in \eqref{CLZ1}, the one involving $\alpha$,  is a piecewise smooth function.  Thus, there are no non-trivial solutions of this type for \eqref{CLZ1}, but, by contrast, these solutions are supported by \eqref{jc1}.  In \cite{CI}, the system \eqref{hs2b} was considered as the Hunter-Saxton  limit of \eqref{CLZ1}, with vanishing boundary conditions. 

There is another important difference between 
 \eqref{jc1} and \eqref{CLZ1}, which is in the distinct nature of their symmetry hierarchies. All the systems in the list have a trivial symmetry of order 1, corresponding to translation in $x$. The lowest non-trivial symmetry of the system  \eqref{jc1}  is of order 3 in derivatives of $m_1$ and $m_2$, whereas \eqref{CLZ1} (being of nonlinear Schr\"odinger type) has a symmetry of order 2. 
\end{remark}

\section{ The Neumann beam problem and Hunter-Saxton limit} \label{sec:NB/HS}
\setcounter{equation}{0}

In this section we construct a specific class of weak solutions of the generalized beam problem. 
For \eqref{vecsys}, we consider the Clifford-valued Hunter–Saxton limit ($\alpha=0$) in Euclidean signature.  
Let $(\cdot,\cdot)$ be a nondegenerate symmetric bilinear form on $W$ with Euclidean signature, and decompose
\[
W = V \oplus V^\perp, \qquad V = \langle e_\mu\rangle_{2\le \mu\le d}, 
\quad V^\perp = \langle e_1\rangle.
\]
Denote by $P_V$ the orthogonal projection onto $V$.

If we rewrite \eqref{vecsys} in the form \eqref{mtcliff} as in section 2, 
with $b_0=u$ and $b_{-1}=e_1$, 
then the measure $M$ evolves according to
\begin{equation} \label{eq:Mtgen_compact}
M_t = \mathcal{L}_M u + (M,\bv)e_1 - (M,e_1)\bv, \qquad \bv\in V,
\end{equation}
with the operator 
$$ 
\mathcal{L}_M =MD_x+D_xM 
$$
(which acts the same as $ \tfrac{1}{2}(A_M D_x + D_x A_M) $ on the scalar field $u\in\F$), 
subject to
\begin{equation} \label{eq:uvconstraints_compact}
u_{xx}=2(M,e_1), \qquad \bv_x=-2P_V M.
\end{equation}
The right-hand side of \eqref{eq:Mtgen_compact} splits into the component parallel to $M$ and a perpendicular rotation term.

Eliminating $M$ yields the closed system
\begin{align}
u_{xxt} &= (uu_{xx})_x + \tfrac12 (u_x^2)_x - \tfrac12 (\bv,\bv)_x, \label{eq:uxxt_compact}\\
\bv_{xt} &= (u\bv)_{xx}, \label{eq:vxt_compact}
\end{align}
which corresponds to \eqref{hs2a} when $d=2$, 
or, after one integration of both equations, we have 
\begin{align}
u_{xt} &= uu_{xx} + \tfrac12 u_x^2 - \tfrac12 (\bv,\bv), \\
\bv_t &= (u\bv)_x, 
\end{align}
which corresponds to \eqref{hs2b} when $d=2$. 
The reduction $\bv=0$ recovers the scalar Hunter--Saxton equation \cite{HS, HZ}.  In the rest of the paper, we work with the unintegrated system, given by \eqref{eq:uxxt_compact} and \eqref{eq:vxt_compact}. 
The system is $O(V)$–invariant: $u$ is fixed and $\bv\mapsto g\,\bv$ for $g\in O(V)$.


We represent the fields in terms of $M$ by
\begin{align}
u(x) &= \int_\mathbb{R} |x-y|(M(y),e_1)\,\rd y, \label{eq:uintf} \\
\bv(x) &= -P_V \int_\mathbb{R} \operatorname{sgn}(x-y)M(y)\,\rd y,
\end{align}
which is natural for compactly supported measures. (For convenience, the implicit dependence on time $t$ is omitted here.)

\begin{lemma}[Asymptotics and mass conservation] \label{lem:uvass} 
Let $\mathcal M=\int M(y,t)\, \rd y$ and $\mathcal M_1=\int yM(y,t)\,\rd y$.  
Then outside the support, the fields behave as 
\begin{align}
u_+(x) &=  (\mathcal M,e_1)x-(\mathcal M_1,e_1), \qquad u_-=-u_+, \label{eq:up}\\
\bv_\pm &= \pm P_V\mathcal M, \label{eq:vpm}
\end{align}
where $\pm$ refers to the right $(+)$, left $(-)$, asymptotic regions, respectively. Moreover, we have conservation of mass, in the sense that   
$$ 
\frac{d}{dt}\mathcal M=0.
$$
\end{lemma}


\begin{proof}
Integrating \eqref{eq:uxxt_compact} and \eqref{eq:vxt_compact} over $\mathbb R$ and using the asymptotics shows that all boundary contributions vanish.
\end{proof}


\subsection{Time evolution and the Weyl function}

For \eqref{Phixx} with $\al=0$, we consider the Neumann string problem in $\Cl(W)$:

\begin{definition}\label{def:genN} 
For compactly supported $M$, the Neumann problem is
\begin{equation}
\Psi_{xx}=\lambda M\Psi,
\end{equation}
with $\Psi_x$ non-invertible at both spatial infinities.
\end{definition}

We specialize to a discrete measure
\begin{equation} \label{eq:fM}
M=\sum_{j=1}^N M_j\delta_{x_j}, \qquad x_1<\cdots<x_N,
\end{equation}
with invertible $M_j\in \Cl_1$.
Then \eqref{eq:Mtgen_compact} reduces to
\begin{align}
\dot x_j &= -u(x_j), \label{eq:dotxj}\\
\dot M_j &= \langle u_x\rangle(x_j)M_j
+\tfrac12[e_1\langle \bv\rangle(x_j),M_j]. \label{eq:dotMj}
\end{align}
Invertibility of each $M_j$ for all $t$, provided that $M_j(0)$ is invertible, and conservation of $\mathcal M=\sum_j M_j$ follow directly from the above ODE system.


To incorporate time deformation into the spectral problem, we use \eqref{Phit} in the form 
\begin{equation}
\Psi_t = \tfrac12(e_1\bv-u_x)\Psi + \left(u+\frac{e_1}{\lambda}\right)\Psi_x.
\end{equation}
Let $\Phi$ be the solution normalized by $\Phi=1$ for $x\ll0$.  
A time-dependent gauge $\Omega(t,\lambda)$ is needed to ensure compatibility. Thus, setting 
\[
\Psi=\Phi\Omega, 
\]
and analyzing the compatibility condition in the left asymptotic region, one obtains
\begin{equation} \label{eq:omegadef}
\Omega_t\Omega^{-1}
=\tfrac12(e_1P_V\mathcal M + (\mathcal M,e_1)),
\end{equation} 
which is constant in $t$.
In the right asymptotic region, writing
\beq\label{rightasy}
\Phi e_1=A(x-x_N)+B,
\eeq 
the coefficients evolve as follows.

\begin{proposition}
Let $F=\frac12 e_1P_V\mathcal M$. Then $A,B$ in \eqref{rightasy} satisfy 
\begin{align}
A_t+[F,A]&=0, \label{eq:At} \\
B_t+(\mathcal M,e_1)B+[F,B]&=\frac{e_1}{\lambda}A. \label{eq:Bt} 
\end{align}
\end{proposition}
\begin{proof} 
Using \eqref{eq:omegadef} and 
\autoref{lem:uvass},  we obtain
\begin{equation}\label{eq:Phit}
\dot{\Phi} +\Phi\omega+\omega\Phi =(\overbrace{(\mathcal{M},e_1)x-(\mathcal{M}_1, e_1)}^{u_+}+\frac{e_1}\lambda)\Phi_x
\end{equation}
where $\omega=\frac{e_1 P_V\mathcal{M}+(\mathcal{M}, e_1)}{2}$.  
The term $\Phi\omega+\omega \Phi$ can be easily simplified to
\begin{equation*}
\Phi\omega+\omega \Phi=(\mathcal{M}, e_1) \Phi + \Phi \frac{e_1 P_V \mathcal{M}}{2} +\frac{e_1 P_V \mathcal{M}}{2}\Phi.  
\end{equation*}
We now observe that for any $\bv\in V$, $e_1\bv+\bv e_1=0$, which results in the formula 
\begin{equation*}
\Phi\omega e_1+\omega \Phi e_1=(\mathcal{M}, e_1) \Phi e_1 + \frac{1}{2}\left[e_1 P_V \mathcal{M},\Phi e_1\right].  
\end{equation*}
Finally, multiply \eqref{eq:Phit} by $e_1$ on the right and write $\Phi e_1=A(x-x_N)+B$, recalling that $\dot x_N=-u_+(x_N)$, to get \eqref{eq:At} and \eqref{eq:Bt}.  
\end{proof} 
\begin{corollary}\label{cor:ABintegrated} 
The time dependence of $A$ and $B$ is 
\begin{align}
A(t)&=e^{-Ft}A(0)e^{Ft}, \\
B(t)&=e^{-(\mathcal M,e_1)t}e^{-Ft}B(0)e^{Ft}
+\frac{e_1}{\lambda(\mathcal M,e_1)}(1-e^{-(\mathcal M,e_1)t})A(t),
\end{align}
from which it follows that the Weyl function
\begin{equation} \label{eq:W} 
W(\lambda)
=B A^{-1}
\end{equation} 
has time evolution given explicitly by
\begin{equation}  \label{eq:Weyl evol}
W(\lambda,t)
=e^{(\mathcal M,e_1)t}e^{-Ft}W(\lambda,0)e^{Ft}
+\frac{e_1}{\lambda(\mathcal M,e_1)}(1-e^{-(\mathcal M,e_1)t}).
\end{equation}
\end{corollary}

Observe that in the asymptotic region $x>>0$,  $A(t)=\Phi_x e_1$. \autoref{cor:ABintegrated}
provides a plethora of conserved, non-local quantities in $M$.  
\begin{corollary}\label{cor:tPhix}
 Let $x$ be in the right asymptotic region.    
 Then 
 \begin{equation}
  e^{F t}\Phi_xe_1 e^{-Ft}  
 \end{equation}
 is conserved.
\end{corollary}

\subsection{Low-order $\lambda$ expansion} 
The wave function $\Phi$ is, in general, an entire function of $\lambda$.  In this paper, we will only use the first few terms in its expansion about 
$\lambda=0$.  We now examine the initial terms of this expansion.  
We construct $\Phi$, subject to the Neumann condition for $x<<0$, by solving iteratively.  To this end, we write the expansion around $\la=0$ as %
\begin{equation}
\Phi=1+\lambda\Phi^{(1)}+\lambda^2 \Phi^{(2)}+O(\lambda^3), 
\end{equation} 
where the iterates satisfy 
\begin{equation} 
\Phi_{xx}^{(n)} =M\Phi^{(n-1)}, \qquad
\mathrm{with} \quad \Phi^{(n)}(x)=0, \quad \Phi_{x}^{(n)}(x)=0, \quad \text{  for  } x<<0,  \text{   and all  } n\geq 1.  
\end{equation}
In a routine way, one gets the recurrence 
\begin{equation} 
\Phi^{(n)}(x)=\int_{-\infty<z<x} (x-z)M(z) \Phi^{(n-1)}(z) \,\rd z, 
\end{equation}
whose solution reads 
$$ 
\Phi^{(n)}(x)= 
\int\limits_{-\infty<z_1<z_2<\cdots<z_n<x} 
M(z_n)M(z_{n-1})\cdots M(z_1)\, \rd z_1\rd z_2\cdots \rd z_n.  
$$ 
We will only need the derivatives of $\Phi^{(1)}$ and $\Phi^{(2)}$ in the asymptotic region $x>>0$.  The following integral formulas provide them as
\begin{align} 
&\Phi_{x}^{(n)}=\int_{-\infty}^{\infty} M(z_1) \rd z_1=\mathcal{M}, \\
&\Phi_{x}^{(2)}=\int_{-\infty}^\infty M(z_2)\int_{-\infty}^{z_2} (z_2-z_1) M(z_1) \rd z_1\, \rd z_2.  \label{eq:Phi2}
\end{align} 

We can now infer the time evolution of \eqref{eq:Phi2} in the asymptotic region $x>>0$.  Indeed, by \autoref{cor:tPhix}, 
\begin{equation} 
\Phi_{x}^{(2)}=e^{-Ft} (c_1+c_3)e^{Ft}e_1^{-1}, 
\end{equation} 
where $c_1, c_3$ are constant elements of $\Cl_1, \Cl_3$,  respectively.  Of particular interest to us is, in the region $x>>0$, the trace of $\Phi_{x}^{(2)}$ which is given by 
\begin{equation} 
\tr(\Phi_{x}^{(2)})=\tr\big((c_1+c_3)e^{Ft}e_1^{-1} e^{-Ft}\big).  
\end{equation}  
\begin{proposition} \label{prop:Frotation}
When the bilinear form $(\cdot,\cdot)$ has Euclidean signature 
\begin{equation} 
 e^{Ft}e_1^{-1} e^{-Ft}=\cos (\theta t) \, e_1^{-1}- \frac{\sin (\theta t)}{\theta} \, P_V\mathcal{M},
 \end{equation} 
 where $\theta=||(P_V\mathcal{M}||$.  
 \end{proposition} 
 \begin{proof}
 Since $e_1^{-1}=e_1$, we might as well work with  $Z(t)=e^{Ft}e_1 e^{-Ft}$.  Recall that $F=\frac{e_1 P_V \mathcal{M}}{2}$.  Hence,  $e^{Ft}e_1 e^{-Ft}$ represents a rotation  in the 
 $(e_1, P_V\mathcal{M})$ plane.  So $Z(t)=a(t)e_1+b(t) P_V\mathcal{M}$, and 
 $$\dot Z=\dot a e_1+\dot b P_V\mathcal{M}. $$  On the other hand 
 $$\dot Z=[\frac{e_1P_V \mathcal{M}}{2}, ae_1+b P_V\mathcal{M}]=-a(e_1,e_1)P_V\mathcal{M}+b(P_V\mathcal{M},P_V\mathcal{M})e_1. $$
 Comparing the coefficients of $e_1, P_V\mathcal{M}$, we get a system of ODEs
 \begin{equation*} 
 \dot a=||P_V\mathcal{M}||^2 b, \qquad \dot b=- a, 
 \end{equation*} 
 whose solution is 
 $$ 
 a(t)=\cos(||P_V\mathcal{M}|| \, \,  t), 
 \qquad b(t)=- \frac{\sin( ||P_V \mathcal{M}|| \, \,  t)}{||P_V\mathcal{M}||}. 
 $$ 
 %
 \end{proof} 
 \begin{corollary} \label{cor:trace evolution}
 In the region $x>>0$, the trace of $\Phi_{x}^{(2)}$ takes the form
\begin{equation} 
\tr(\Phi_{x}^{(2)})=C_1 \cos (\theta t)+ C_2 \sin (\theta t), 
\end{equation} 
for some constants $C_1, C_2$,  and $\theta$ defined in \autoref{prop:Frotation}.  
\end{corollary} 


\subsection{ The forward problem}
We start by solving  the initial value problem 
\begin{equation}\label{eq:IVP}
\Phi_{xx}=\lambda M \Phi, \qquad \mathrm{with} \quad \Phi(x)=1, \,\,\Phi_x(x)=0 \quad\mathrm{for} \quad x<x_1,   \end{equation}
for the Clifford value measure specified in \eqref{eq:fM}.  
The solution $\Phi$ is piecewise linear in $x$, while $\Phi_x$ is piecewise constant.  Thus, for any $\lambda$, the solution is
uniquely characterized by $\Phi_j=\Phi(x_j)$, the left derivatives $\Phi'_j=\Phi_x(x_j-)$, the jump condition 
\begin{equation} \label{eq:Phixjump}
\Phi'_{j+1}-\Phi_j=\lambda M_j \Phi_j, \qquad 1\leq j\leq N, 
\end{equation}
and the relation
\begin{equation}\label{eq:Phirec}
\Phi_{j+1}=\Phi_{j}+l_j \Phi'_{j+1}, \qquad 1\leq j\leq N-1. 
\end{equation}
We note that $\Phi'_{N+1}$ is just the right derivative 
$\Phi_x(x_N+)$.  Both recurrence relations can be written 
in matrix form, as follows: 
\begin{subequations}
\begin{equation} \label{eq:firstrec}
\hspace{0.55cm} \begin{bmatrix}\Phi_{j+1}\\ \Phi'_{j+1}\end{bmatrix}=\begin{bmatrix}
    1&l_j\\ 0&1 \end{bmatrix} \begin{bmatrix}\Phi_j\\ \Phi'_{j+1}
\end{bmatrix}, \qquad 1\leq j\leq N-1, 
\end{equation}
\begin{equation} \label{eq:secondrec}
\begin{bmatrix}\Phi_{j}\\ \Phi'_{j+1}\end{bmatrix}=\begin{bmatrix}
    1&0\\ \lambda M_j&1 \end{bmatrix} \begin{bmatrix}\Phi_j\\ \Phi'_{j}
\end{bmatrix}, \qquad 1\leq j\leq N. 
\end{equation}
\end{subequations}

\begin{proposition}\label{prop:inv} 
The quantities $\Phi_{j}$ and $\Phi_{j}'$ are polynomials in   $\lambda$ with degrees 
$\deg(\Phi_{j})=\deg(\Phi'_j)=j-1$. 
For all $j\geq 2$, $\Phi_j$ and $\Phi'_j$ are invertible for sufficiently large
$\lambda$.  
\end{proposition}
\begin{proof}
The proof proceeds by induction on $j$. 
The fact that they are polynomials of 
the given degrees  follows directly from the recurrence relations. The base case 
corresponds to $j=2$.  We have 
\begin{equation}
\begin{bmatrix}\Phi_2\\ \Phi_2' \end{bmatrix}=\begin{bmatrix}1+\lambda l_1 M_1&l_1\\ \lambda M_1&1 \end{bmatrix}\begin{bmatrix}1\\0
\end{bmatrix}=\begin{bmatrix}1+\lambda l_1 M_1\\\lambda M_1 \end{bmatrix},   
\end{equation}
where we used $\Phi_1=1$ and $\Phi'_1=0$ from the initial value problem \eqref{eq:IVP}. 
Both entries are invertible for large $\lambda$ because $M_1$ is 
invertible.  
The recurrence relations give the following:
\begin{equation*}
\begin{bmatrix}\Phi_{j+1}\\ \Phi_{j+1}' \end{bmatrix}=\begin{bmatrix}1+\lambda l_j M_j&l_j\\ \lambda M_j&1 \end{bmatrix}\begin{bmatrix}\Phi_j\\\Phi'_j
\end{bmatrix}=\begin{bmatrix}\Phi_j+\lambda l_j M_j\Phi'_j\\\Phi'_j+\lambda M_j\Phi_j \end{bmatrix}.  
\end{equation*}
The dominant terms in $\lambda$ are $\lambda l_j M_j \Phi'_j$ for the first entry, $\lambda M_j \Phi_j$ for the second entry.  
Both are invertible by the inductive hypothesis, so $\Phi_{j+1}$ and $\Phi'_{j+1}$ are invertible.  
\end{proof}
We can now summarize the recursive way of solving 
the initial value problem.  
\begin{proposition} The solution of the initial value problem \eqref{eq:IVP} can be expressed by the recursive formula 
\begin{equation}
\begin{bmatrix}\Phi_N\\ \Phi'_{N+1}\end{bmatrix}=
\begin{bmatrix}1&0\\ \lambda M_N&1 \end{bmatrix}\left( \prod_{j=1}^{N-1}\begin{bmatrix}1+\lambda l_j M_j&l_j\\ \lambda M_j&1 \end{bmatrix}\right)\begin{bmatrix}1\\ 0 \end{bmatrix}
\end{equation}
\end{proposition}
From the definition of the Weyl function given by \eqref{eq:W}, we get a formula 
$W(\lambda,t)=B(\lambda,t)A^{-1}(\lambda,t)=\Phi_Ne_1( \Phi_{N+1}'e_1)^{-1}=\Phi_N( \Phi_{N+1}')^{-1}$.  

\begin{theorem}[Stieltjes-Clifford fraction] \label{thm:SCfrac}
For large $\lambda$, the Weyl function is given by a finite continued fraction of Stieltjes type, with Clifford algebra-valued partial quotients, that is   
\begin{equation} \label{eq:SCfrac}
W(\lambda)=\cfrac{1}{\lambda M_N +\cfrac{1}{l_{N-1}+\cfrac{1}{\lambda M_{N-1}+\cfrac{1}{\ddots+\cfrac{1}{\lambda M_1}}}}}
\end{equation}
\end{theorem}
\begin{proof}
We note that by \autoref{prop:inv} all $\Phi_j$ and $\Phi'_j$ are invertible for sufficiently large $\lambda$.  By the last iteration of each recurrence relation  ($j=N$ in \eqref{eq:Phixjump} and 
\eqref{eq:Phirec}), we have  
\begin{align*}
&\Phi'_{N+1}=\Phi'_N+\lambda M_N \Phi_N
, \\
&\Phi_N=\Phi_{N-1}+l_{N-1} \Phi'_N,  
\end{align*}
which gives 
\begin{align*}
    \Phi_N(\Phi'_{N+1})^{-1}=&(\lambda M_N+\Phi'_N \Phi_N^{-1})^{-1}=
    (\lambda M_N+(\Phi_N(\Phi'_N)^{-1})^{-1})^{-1}\\=&\big(\lambda M_N+(l_{N-1}+\Phi_{N-1} (\Phi'_{N})^{-1})^{-1}\big)^{-1}.  
\end{align*}
This is the first step in the iteration; we moved from $\Phi_N(\Phi'_{N+1})^{-1}$ to $\Phi_{N-1}(\Phi'_{N})^{-1}$.  We proceed by iterating down until we reach $\Phi_1(\Phi'_{2})^{-1}=\lambda M_1$, since
$\Phi_1=1$ and $\Phi'_1=0$ for the Neumann problem.   This concludes the proof.  
\end{proof}

Finally, the time evolution of $W$ is given by \eqref{eq:Weyl evol}, while the continued fraction expansion of $W$ determines the masses $M_j$ and the spacings $l_{N-1}, \dots, l_1$. 
We will provide a detailed analysis of the inverse problem elsewhere, with particular emphasis on the explicit representations of the quantities \(M_j\) and \(l_j\) in terms of the Weyl matrix function \(W\). 
%
To complete the formal solution of the inverse problem, we still need to determine the position of the last particle, $x_N$. 
\subsection{Determining the position $x_N$}
We cannot determine the position $x_N$ from the Weyl function $ W$. This is similar to the scalar Hunter-Saxton equation, for which the same question was addressed in  \cite{BSS2001}.  
$x_N$ satisfies (see \eqref{eq:dotxj})
\begin{equation*} 
\dot x_N=-u(x_N)=-u_+(x_N)=-((\mathcal{M}, e_1)x_N-(\mathcal{M}_1,e_1)), 
\end{equation*}
or, 
\begin{equation}\label{eq:dotxN}
\dot x_N+(\mathcal{M},e_1)x_N=(\mathcal{M}_1,e_1). 
\end{equation} 
The latter equation can be easily integrated, provided we know the time dependence of $(\mathcal{M}_1,e_1)$.  
\begin{proposition} The inner product of the first moment $\mathcal{M}_1$ with $e_1$ satisfies 
\begin{equation} 
\frac{d}{dt} (\mathcal{M}_1, e_1)=-\tr(\Phi_{x}^{(2)}).  
\end{equation} 

\end{proposition}
\begin{proof} 
We want to compute $\frac{d}{dt} \sum_j x_j (M_j, e_1)$.  
Since $\dot x_j=-u(x_j)$ and $\dot M_j=\avg{u_x}(x_j)M_j +\frac12[e_1\avg{v}(x_j), M_j]$ we obtain 
\begin{equation} 
\frac{d}{dt} \sum_j x_j (M_j, e_1)=\overbrace{\sum_j (-u(x_j)+x_j\avg{u_x}(x_j))(M_j, e_1)}^{=I_1}+\overbrace{\sum_j x_j (\avg{v}(x_j), M_j) (e_1, e_1)}^{=I_2}
\end{equation} 
We recall that $u(x)=\int_\R \abs{x-y}(M(y), e_1) \rd y$ which in the discrete case becomes
\begin{equation*} 
u(x_j)=\sum_k \abs{x_j-x_k} (M_k, e_1). 
\end{equation*}
Likewise, $\avg{u_x}=\int_\R \sgn(x-y) (M(y), e_1) \rd y$.  Hence, 
\begin{equation*} 
\avg{u_x}(x_j)=\sum_k \sgn(j-k) (M_k, e_1). 
\end{equation*} 
After symmetrizing, this gives 
\begin{equation*} 
\sum_j x_j\avg{u_x}(x_j)=\frac12 \sum_{j,k} \abs{x_j-x_k} (M_k, e_1)(M_j, e_1).  
\end{equation*} 
Thus we have 
\begin{equation} \label{eq:I1} 
 I_1=-\frac12 \sum_{j, k}\abs{x_j-x_k} (M_k, e_1)(M_j, e_1). 
\end{equation} 
To compute $I_2$ we need to evaluate $(\avg{v}(x_j), M_j)$.  Recall that 
\begin{equation*} 
v(x)=-P_V \int_\R \sgn(x-y)M(y) \rd y. 
\end{equation*} 
Since $P_V M=M-P_{V^\perp} M=M-(M,e_1)e_1$ we get 
\begin{equation*} 
\avg{v}(x_j)=-\sum_k \sgn(j-k)M_k+\sum_k \sgn(j-k) (M,e_1)e_1, 
\end{equation*} 
and, 
\begin{equation*} I_2=-\sum_{j,k} x_j \sgn(j-k) (M_k,M_j) (e_1, e_1)+ \sum_{j,k}x_j\sgn(j-k)(M_k,e_1)(M_j,e_1)
\end{equation*} 
After symmetrizing, we get
\begin{equation} \label{eq:I2}
I_2=-\frac12 \sum_{j,k} \abs{x_j-x_k}  (M_k,M_j) (e_1, e_1)+ \frac12\sum_{j,k}\abs{x_j-x_k} (M_k,e_1)(M_j,e_1)
\end{equation} 
Combining \eqref{eq:I1} and \eqref{eq:I2} 
we get 
\begin{equation} 
 I_1+I_2=-\frac12 \sum_{j,k} \abs{x_j-x_k} (M_k, M_j).  
\end{equation} 
For comparison, let us compute $\tr{\Phi_x^{(2)}}$ in the right asymptotic region assuming the mass is discrete.  
From \eqref{eq:Phi2} we get  
\begin{equation*} 
\Phi_{x}^{(2)}=\sum_{j<k} (x_k-x_j)M_kM_j, 
\end{equation*}
and 
\begin{align*} 
\tr{\Phi_{x}^{(2)}}=\sum_{j<k} (x_k-x_j)\tr{(M_kM_j)}= \sum_{j<k} (x_k-x_j) (M_k, M_j)\\=\frac12(\sum_{j<k} (x_k-x_j)(M_k, M_j)+\sum_{j>k}
(x_j-x_k)(M_j,M_k)).  
\end{align*} 
Hence, 
\begin{equation} 
\tr{\Phi_{x}^{(2)}}=\frac12 \sum_{j,k} \abs{x_k-x_j}(M_k,M_j)
\end{equation} 
which implies the final claim.  
\end{proof} 
This result, together with  \autoref{cor:trace evolution}, allows one to integrate \eqref{eq:dotxN}.  Let us set $(\mathcal{M}_1,e_1)=c_1\cos (\theta t)+c_2\sin (\theta t)$.  
Then the solution $x_N$ reads as follows: 
\begin{equation} \label{eq:xN}
x_N(t) \! 
=c_{3}{\mathrm e}^{-(\mathcal{M}, e_1) t}  +\frac{(c_1 \left(\mathcal{M},e_1) -c_2 \theta \right)\cos \! \left(\theta  t \right) +(c_1\theta  +c_2 (\mathcal{M}, e_1))\sin \! \left(\theta  t \right)}{(\mathcal{M}, e_1)^{2}+\theta^{2}}.  
\end{equation}

This completes the solution of the forward and inverse problems for the generalized Neumann beam in the Hunter--Saxton regime.
\subsection{The case of two peakons with two components} \label{sec:HS2/peakons}

Here we can specialize the results above to the case $N=2$, to understand the interaction of two peakons, and further restrict to the two-component Hunter-Saxton case $d=2$, which we refer to as HS2, given by \eqref{hs2a}. 
For any $d$, 
in the special case of $N=2$, 
the wave $u(x,t)$ defined by \eqref{eq:uintf} takes the form
\begin{equation}\label{u2peakon}
u(x,t) = m_{11}(t)\,|x - x_1(t)| + m_{21}(t)\,|x - x_2(t)|,
\end{equation}
where the coefficient functions are defined by 
$m_{j\mu}(t)=\big(M_j(t), e_\mu\big)$ for $ j=1,\cdots, N, \, \mu=1,\cdots, d$. 

In the case of HS2, with two components ($d=2$), 
the interaction of two peakons is illustrated in Fig.\ref{fig: HS2_figures} below. 
\begin{figure}[htbp]
\centering
\subfloat[Positions of peaks for two HS2 peakons.  The initial conditions are: $x_1(0)=0, x_2(0)=1, m_{11}(0) =1, m_{12}(0) =0.1, m_{21}(0) =-1.5, m_{22}(0) =-10.$]{\includegraphics[width=0.55\textwidth]{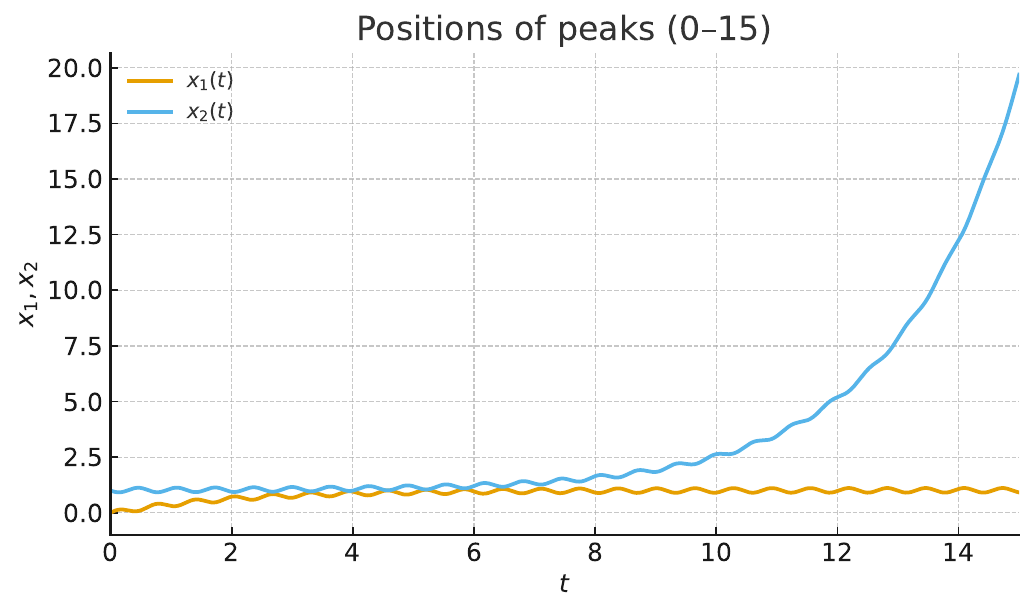}}
\hfill
\subfloat[Amplitudes of peaks for two peakons; $0\leq t\leq 15.$]{\includegraphics[width=0.55\textwidth]{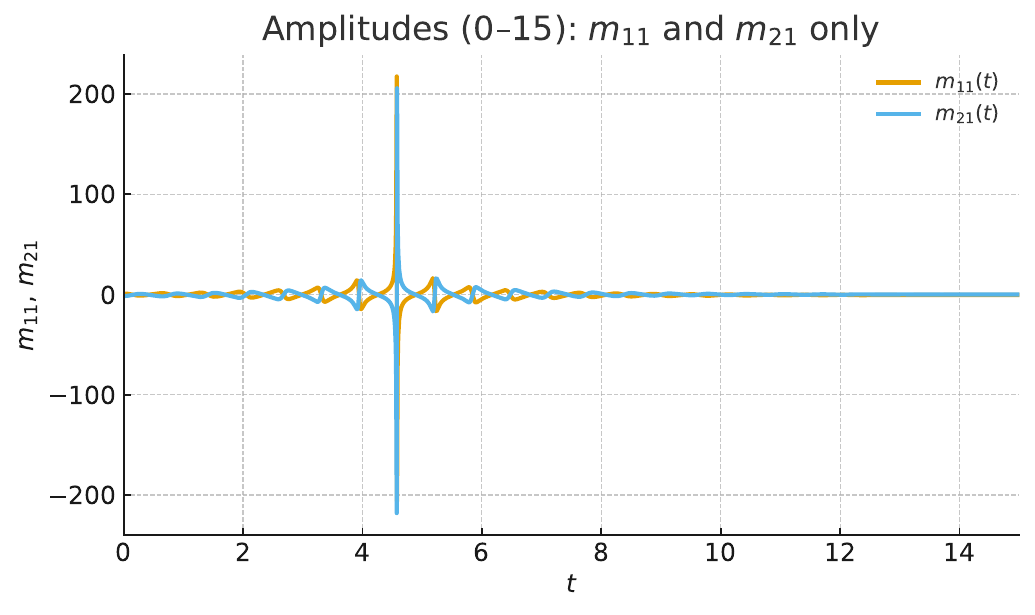}}
\subfloat[Amplitudes of peaks for two peakons; longer times: $15\leq t\leq 30.$]
{\includegraphics[width=0.55\textwidth]{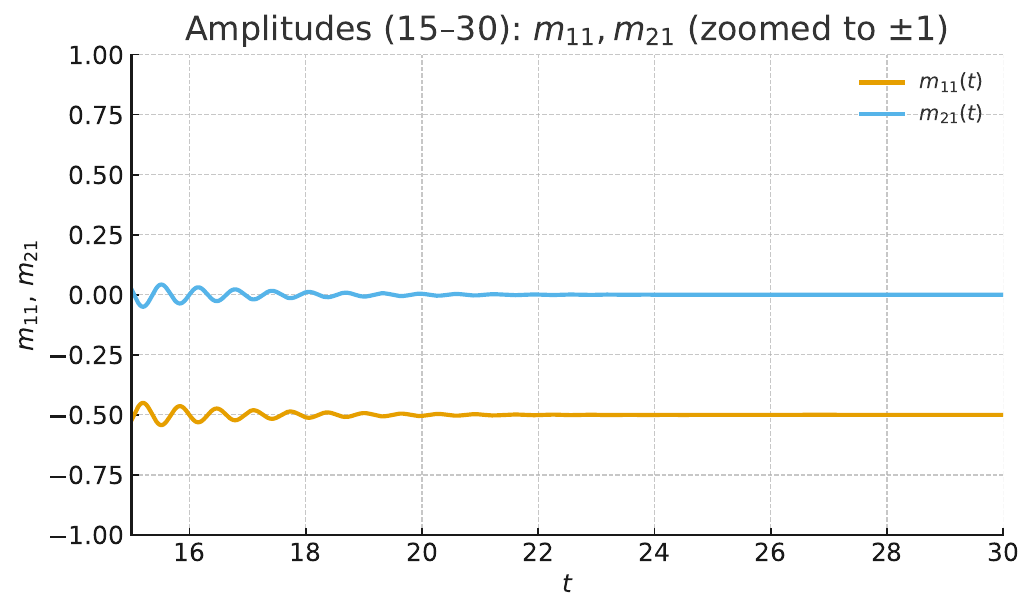}}

\caption{HS2 peakons satisfying \eqref{hs2b}.}
\label{fig: HS2_figures}

\end{figure}

\begin{example} \label{ex:sysHS2}
Let us set $(M_j, e_\mu)=m_{j\mu}$ for $ j=1,\cdots, N, \, \mu=1,\cdots, d$.  
When $N=2, d=2$, the equations of motion \eqref{eq:dotxj} and \eqref{eq:dotMj} read: 
\begin{align} 
&\dot x_1=-m_{21}|x_1-x_2|, \\ 
&\dot x_2=-m_{11}|x_1-x_2|, \\ 
&\dot m_{11}=m_{11}^2-m_{12}^2 -C_1m_{11}+C_2m_{12}, \\ 
&\dot m_{12}=2m_{11}m_{12}-C_1m_{12}-C_2m_{11}, 
\end{align} 
where $m_{11}+m_{21}=C_1$ and $\quad m_{12}+m_{22}=C_2$ are constants of motion.  
The corresponding field $u(x,t)$, satisfying one half of the system \eqref{hs2b}, is given by the formula \eqref{u2peakon}, where the amplitudes  $m_{11}$ and $m_{21}=C_1-m_{11}$ and the peak positions $x_1,x_2$ are obtained from the solution of this system of ODEs. 
The plots in Fig.\ref{fig: HS2_figures} are produced from a particular solution of the latter system, with 
$(\mathcal{M},e_1)=m_{11}+m_{21}=C_1=-0.5$. We can make the following observations: 
\begin{itemize}
  \item \textbf{Wave centered at \( x_1(t) \):}
  \begin{itemize}
    \item The center of the wave  \( x_1(t) \) becomes \emph{asymptotically periodic} as \( t \to \infty \).
    \item  A longer time simulation shows that amplitude $m_{11}$ stabilizes around $-0.5$.  
    
  \end{itemize}

  \item \textbf{Wave centered at \( x_2(t) \):}
  \begin{itemize}
    \item The center of the wave \( x_2(t) \) increases  rapidly while oscillating and tends to \( +\infty \) as \( t \to \infty \).  Since  $(\mathcal{M},e_1)<0$, this behaviour is consistent with \eqref{eq:xN}.  
    \item A longer time simulation shows that $m_{21}$ stabilizes around $0$; the wave component \( m_{21}(t)\,|x - x_2(t)| \) becomes increasingly flat and moves away from the region of interest.
    \item As a result, this component becomes \emph{asymptotically negligible} in the spatial region near \( x_1(t) \).
  \end{itemize}
\end{itemize}
\end{example} 



\begin{remark}
Constantin and Ivanov studied the short-pulse limit of the 2CH system in reference \cite{CI}, noting that although 2CH 
has no peakon solutions, its short-pulse limit \eqref{hs2b}, a two-component 
generalization of the Hunter-Saxton equation, does.  Our analysis differs in two crucial ways.  We present a general approach to the $d$-component generalized Hunter-Saxton equation, rather than the special case $d=2$.  
Even for $d=2$, our analysis considers different boundary conditions: the coupling field $v$ does not vanish at $\pm \infty$, and neither does $u$.  This is consistent with our interpretation of $v$ as an internal degree of freedom, 
akin to a spin variable. 
\end{remark} 
\section{Conclusions}

We have studied a Clifford-algebraic generalization of the Camassa--Holm (CH) hierarchy in Euclidean
signature, in which the scalar CH variable is coupled to internal degrees of freedom taking values in a
Clifford algebra with $d$ generators. This construction breaks the full $O(d)$ symmetry to the residual
$O(d\!-\!1)$ invariance associated with rotations in the subspace orthogonal to the distinguished direction
$e_1$, and can be interpreted as endowing the CH dynamics with ``spin''-type internal modes. For $d=2$, we
used symmetry methods to classify the resulting perturbations of CH within this framework, obtaining the
system analyzed here as well as an additional integrable candidate.

A key point of the paper is that the integrable structure persists in several complementary guises. On the
one hand, the Hunter--Saxton (short-pulse) reduction admits a measure-valued formulation and a Neumann-type
spectral problem whose Weyl function evolves by an explicit isospectral flow. This leads to a concrete
forward/inverse scheme for discrete measures and provides a Stieltjes-type continued fraction expansion for
the Weyl function in the Clifford setting. On the other hand, via a reciprocal transformation, we relate
our model to the Hirota--Satsuma system, and we include a travelling-wave analysis that clarifies how the
internal variables modify the classical CH wave families.

To illustrate the dynamical impact of the internal degrees of freedom, we carried out 
computations for $N$-peakon solutions in the Hunter--Saxton regime, with particular emphasis on the two-atom
case. In the example presented, the solution exhibits a striking asymptotic decomposition: as $t\to\infty$
the field $u(x,t)$ becomes effectively dominated, on any fixed observable spatial window, by a single
spatially localized oscillatory structure centered near $x_1(t)$, while the second component is transported
to $+\infty$ and simultaneously flattens so that its contribution becomes negligible near $x_1(t)$. The
presence of persistent harmonic modes in the amplitudes is a genuinely new feature compared to the scalar
Hunter--Saxton dynamics and appears to be a robust signature of the Clifford coupling.

Several directions emerge naturally from this work:
\begin{enumerate}
  \item \textbf{Physical interpretation.} Identify a compelling physical setting---for instance, in
  fluid models with internal/charged degrees of freedom---in which the Clifford-valued generalization and
  the residual $O(d\!-\!1)$ symmetry arise from first principles.
  \item \textbf{Signature dependence.} Determine how changing the signature of the underlying bilinear
  form affects qualitative behaviour (e.g., boundedness versus growth, oscillatory versus hyperbolic modes),
  and how this is reflected in the spectral data and time evolution.
  \item \textbf{Higher-rank classification.} Extend the $d=2$ symmetry classification to general $d$ and
  clarify which members of the resulting families are integrable, including the role of reciprocal
  transformations and Hamiltonian structures.
  \item \textbf{Peakon and spectral asymptotics.} Develop a systematic description of long-time behaviour
  for multi-peakon solutions (scattering, clustering, and synchronization mechanisms), and relate these to
  the invariants encoded by the Clifford Weyl function and its continued fraction data.
\end{enumerate}
We expect that progress on these problems will sharpen the connection between the spectral theory of
generalized beam/string problems and the nonlinear dynamics of CH-type equations with internal
degrees of freedom. While we have only considered the case of Hunter-Saxton peakons here with 
$\alpha=0$, the spectral analysis of peakons in the full  vector system with $\alpha\neq 0$ is already underway. 
There is a wealth of other coupled CH-type models whose solutions exhibit interesting behaviours,  such as ``waltzing'' peakons \cite{CHIP,ZQM2020}, and we expect that pursuing an approach via spectral theory will shed further light on these as well. 

\vspace{.1in}
\noindent {\bf Acknowledgements: } ANWH and VSN would like to acknowledge the hospitality we received in the Department of Mathematics, University of Saskatchewan, where JS hosted both of us during visits in May/June 2024 and April 2025. JS's research is supported by the Natural Sciences and Engineering Research Council of Canada (NSERC).

\end{document}